%% file: main.tex
\documentclass[10pt, conference]{IEEEtran}
\usepackage{times}
\usepackage{array}
\usepackage{comment}
\usepackage{subfigure}
\usepackage[ruled, linesnumbered, noend]{algorithm2e}
\usepackage{amsmath}
\usepackage{amsfonts}
\usepackage{amssymb}
\usepackage{amsthm}
\usepackage{thmtools}
\usepackage{thm-restate}
\usepackage{mathrsfs}
\usepackage[capitalise]{cleveref}
\usepackage{mathtools}
\usepackage[pdftex]{graphicx}
\usepackage{multirow}
\usepackage{cite}
\usepackage[utf8]{inputenc}
\usepackage[colorinlistoftodos]{todonotes}
\newcommand{\eg}{\textit{e.g.}, }
\newcommand{\ie}{\textit{i.e.}, }

\newcommand{\tab}[1]{Table~\ref{#1}}
\newcommand{\fig}[1]{Fig.~\ref{#1}}
\newcommand{\alg}[1]{Alg.~\ref{#1}}
\newcommand{\sect}[1]{Section~\ref{#1}}
\newcommand{\eqn}[1]{(\ref{#1})}
\DeclarePairedDelimiter{\ceil}{\lceil}{\rceil}

\def\multiset#1#2{\ensuremath{\left(\kern-.3em\left(\genfrac{}{}{0pt}{}{#1}{#2}\right)\kern-.3em\right)}}

\SetKw{Continue}{continue}
\SetKw{Break}{break}
\SetKw{Return}{return}
\theoremstyle{plain}

\declaretheorem[name=Lemma]{lemma}
\begin{document}
\title{Reliable Slicing of 5G Transport Networks with Dedicated Protection\vspace{-0.5em}}

\author{
	\IEEEauthorblockN{
		Nashid Shahriar\IEEEauthorrefmark{1},
		Sepehr Taeb\IEEEauthorrefmark{1},
		Shihabur Rahman Chowdhury\IEEEauthorrefmark{1},
		Mubeen Zulfiqar\IEEEauthorrefmark{1},\\
		Massimo Tornatore\IEEEauthorrefmark{2},
		Raouf Boutaba\IEEEauthorrefmark{1},
		Jeebak Mitra\IEEEauthorrefmark{3}, and
		Mahdi Hemmati\IEEEauthorrefmark{3}
	}
	\IEEEauthorblockA{
		\IEEEauthorrefmark{1}David R. Cheriton School of Computer Science, University of Waterloo,\\
		\texttt{\{nshahria | staeb | sr2chowdhury | mubeen.zulfiqar | rboutaba\}@uwaterloo.ca}}
		\IEEEauthorrefmark{2}Politecnico di Milano, \texttt{massimo.tornatore@polimi.it}\vspace{0em}\\
	\IEEEauthorrefmark{3}Huawei Technologies Canada Research Center, \texttt{\{jeebak.mitra | mahdi.hemmati\}@huawei.com}
\vspace{-2em}}
\maketitle
\pagestyle{empty}
\input{abstract}
\input{introduction}
\input{related}

\input{model}

\input{ilp-formulation}
\input{heuristic}

\input{evaluation}
\input{conclusion}
\bibliographystyle{IEEEtran}
\bibliography{refdb}
\input{proof}
\end{document}

%% file: abstract.tex
\begin{abstract}
In 5G networks, slicing allows partitioning of network resources to meet stringent end-to-end service requirements across multiple network segments, from access to transport. These requirements are shaping technical evolution in each of these segments. In particular, the transport segment is currently evolving in the direction of the so-called elastic optical networks (EONs), a new generation of optical networks supporting a flexible optical-spectrum grid and novel elastic transponder capabilities. In this paper, we focus on the reliability of 5G transport-network slices in EON. Specifically, we consider the problem of slicing 5G transport networks, i.e., establishing virtual networks on 5G transport, while providing dedicated protection. As dedicated protection requires a large amount of backup resources, our proposed solution incorporates two techniques to reduce backup resources: (i) bandwidth squeezing, i.e., providing a reduced protection bandwidth with respect to the original request; and (ii) survivable multi-path provisioning. We leverage the capability of EONs to fine tune spectrum allocation and adapt modulation format and Forward Error Correction (FEC) for allocating rightsize spectrum resources to network slices. Our numerical evaluation over realistic case-study network topologies quantifies the spectrum savings achieved by employing EON over traditional fixed-grid optical networks, and provides new insights on the impact of bandwidth squeezing and multi-path provisioning on spectrum utilization.

\end{abstract}


%% file: introduction.tex
\section{Introduction}\label{sec:intro}
\begin{figure*}[!t]
  \begin{minipage}[c]{0.59\textwidth}    
  \subfigure[No virtual link demand splitting (100\% BSR)]{\includegraphics[width=0.32\textwidth]{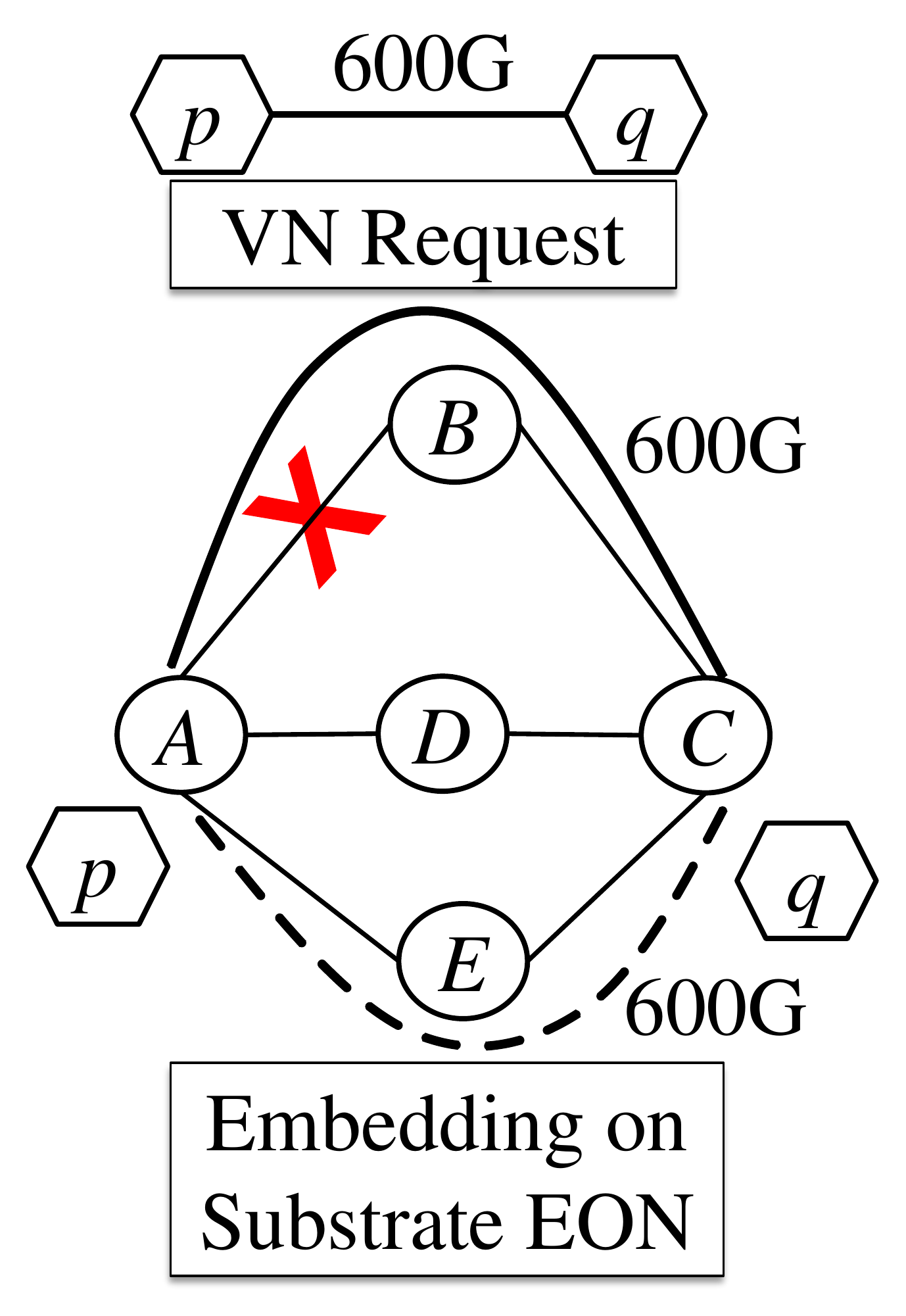}\label{subfig:nosplit}}
  \subfigure[Splitting over disjoint paths (100\% BSR)]{\includegraphics[width=0.32\textwidth]{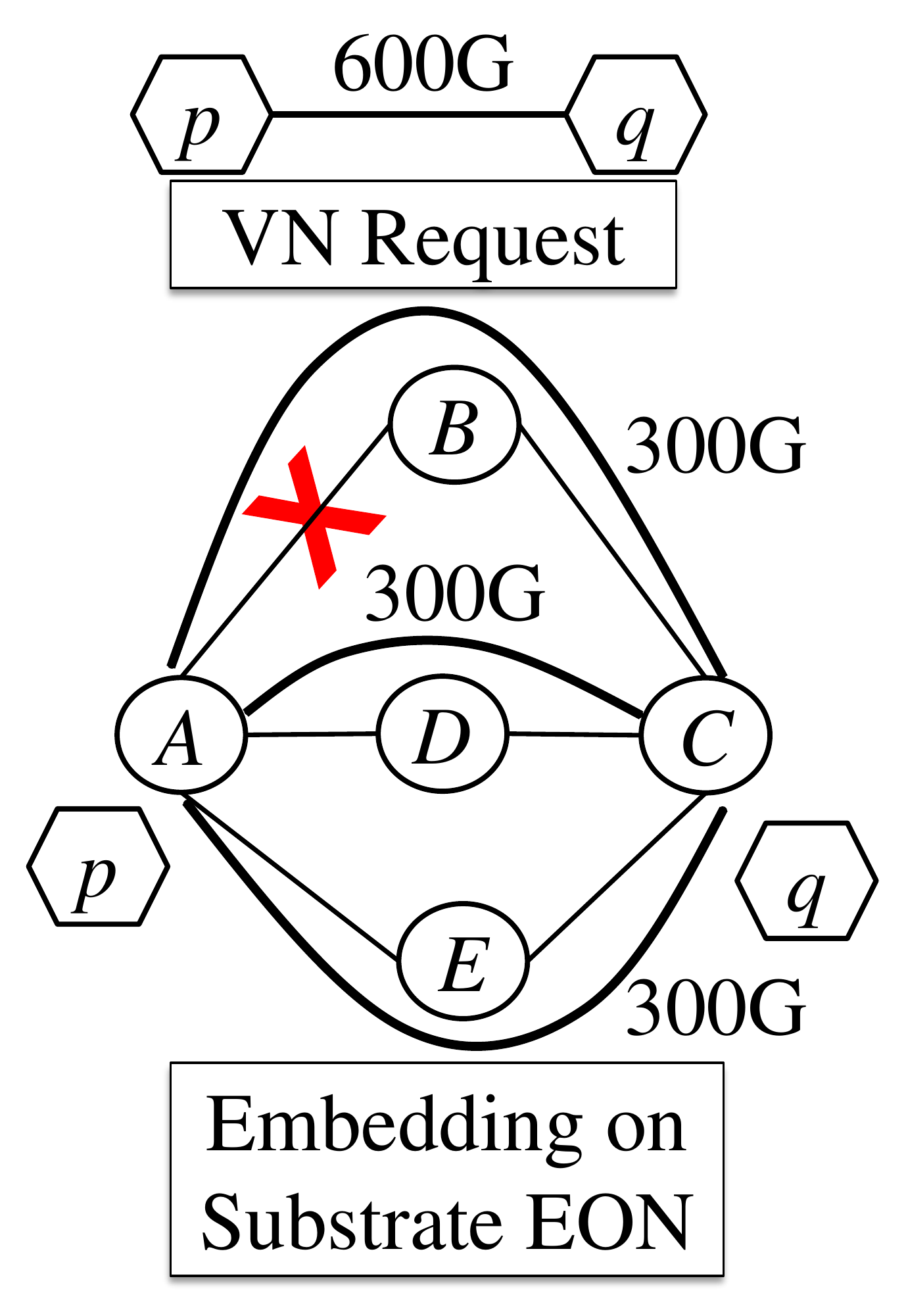}\label{subfig:split}}
  \subfigure[Splitting over disjoint paths (66\% BSR)]{\includegraphics[width=0.32\textwidth]{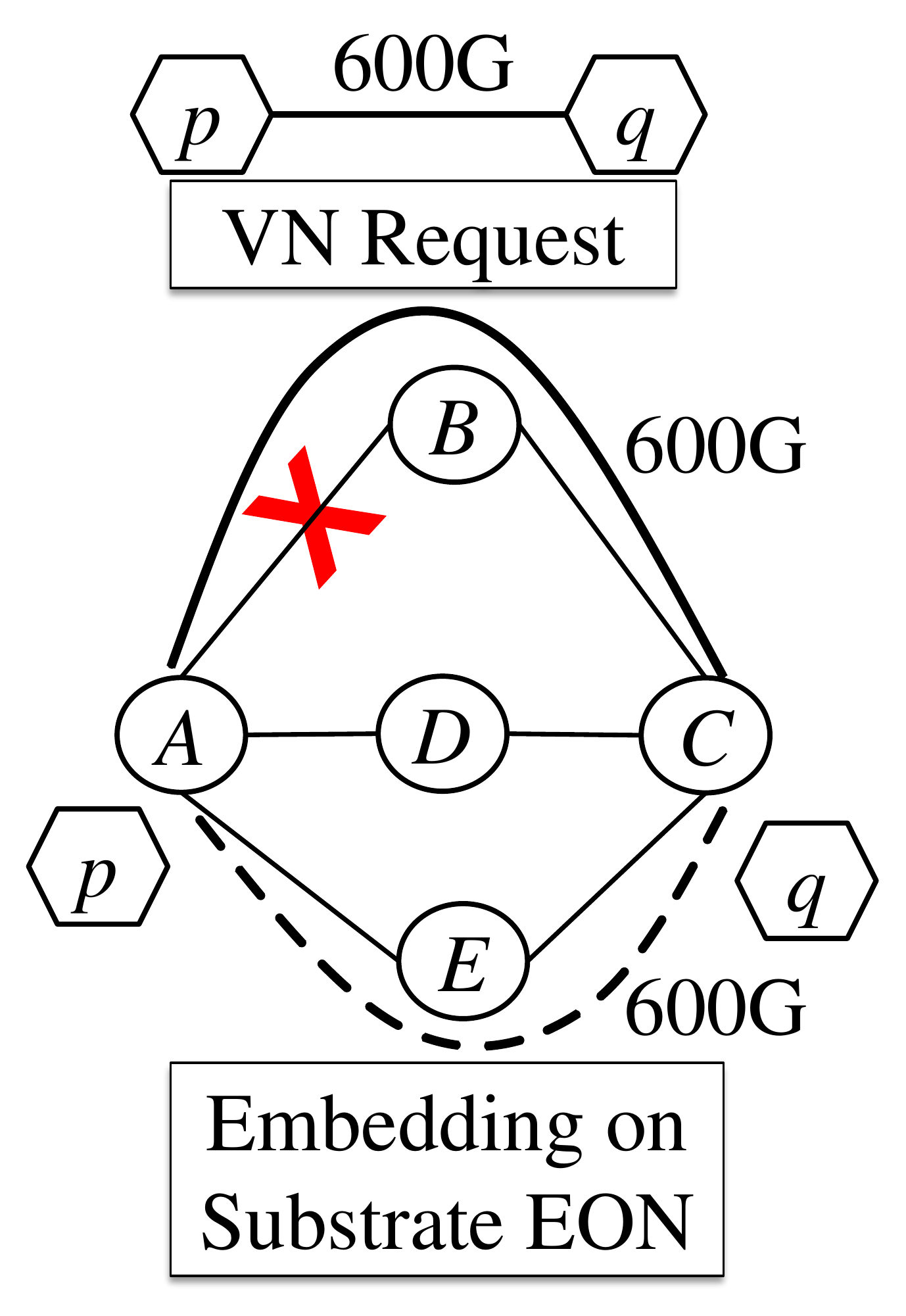}\label{subfig:split-bsr}}
  \end{minipage}
  \begin{minipage}[c]{0.39\textwidth} 
  \subfigure[Splitting over non-disjoint paths (100\% BSR)]{\includegraphics[width=0.49\textwidth]{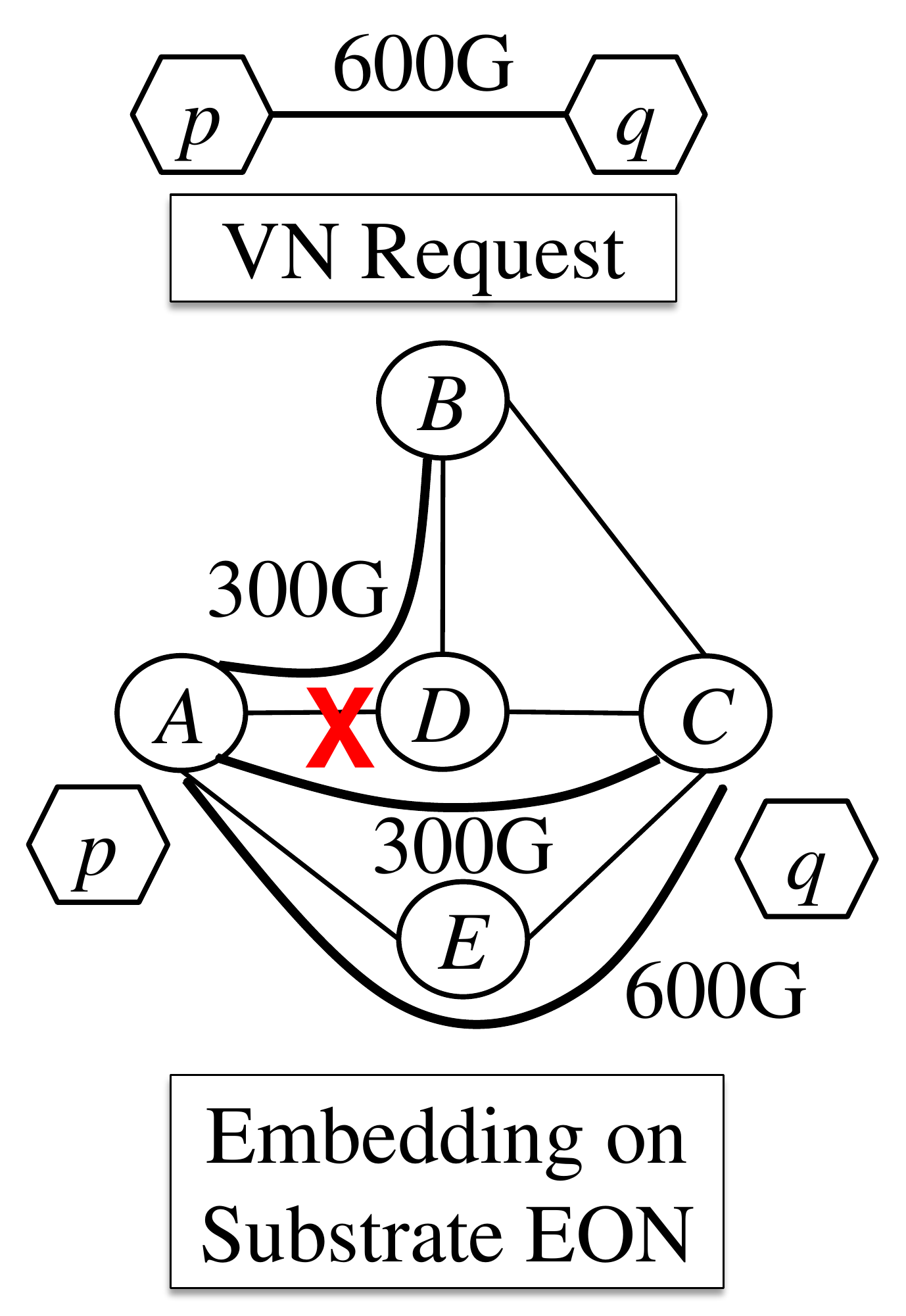}\label{subfig:split-ndp}}
  \subfigure[Splitting over non-disjoint paths (66\% BSR)]{\includegraphics[width=0.49\textwidth]{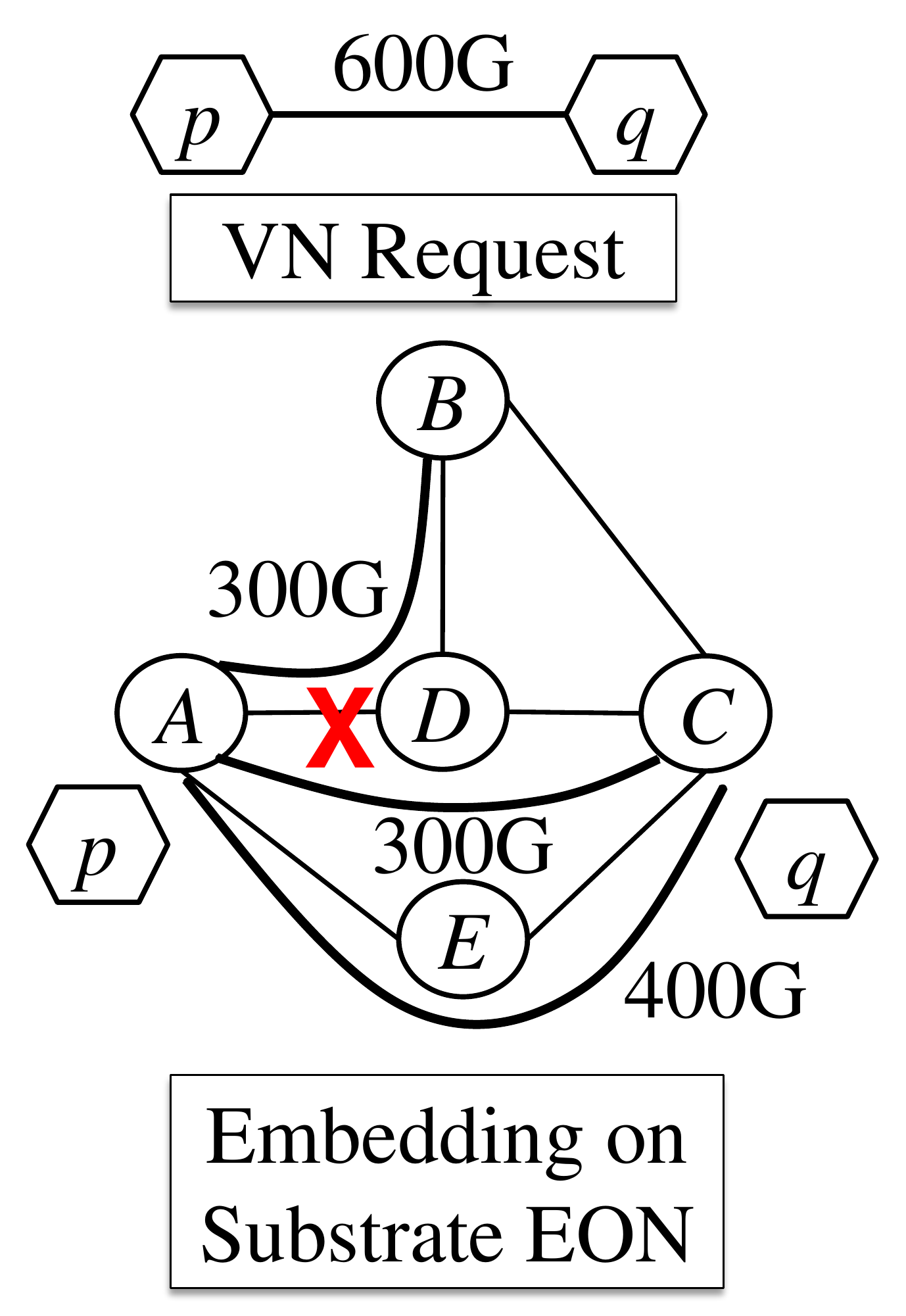}\label{subfig:split-ndp-bsr}}
  \end{minipage}
  \caption{Benefits of virtual link demand splitting and the impact of availability of disjoint paths}\label{fig:split-nodp}
  \vspace{-2em}
\end{figure*}

Transport networks employing the latest advances in elastic optical networking will form the backbone of the fifth-generation (5G) networks and beyond~\cite{gringeri2010flexible,7193532, boutaba2017elastic}. Transport-network capacity must be easily partitionable to facilitate network slicing for 5G services such as enhanced mobile broadband and ultra-reliable low latency communications~\cite{foukas2017network}. Elastic Optical Networks (EONs) are an excellent choice for establishing transport network slices, thanks to their ability to tailor network resources based on service requirements~\cite{hadi2019dynamic}. 

Network slices are usually provided in the form of virtual networks (VNs) (\ie collection of virtual nodes and virtual links (VLinks)). These VNs are anticipated to host services operating at hundreds of Gbps line-rate. As a consequence, even a short-lived network outage can cause a significant traffic disruption in these slices. However, ensuring reliability of network slices is non-trivial as failures in a transport network, such as transponder failure and fiber-cuts, are still the norm. 
A viable solution to ensure slice reliability is to pre-provision dedicated backup paths for each path used to embed the VLinks, also known as dedicated protection. Dedicated backup paths can allow fast fail-over switching within milliseconds~\cite{flexefail}, but they incur a significant resource overhead since they remain idle during failure-free operations


In this paper, we propose a VN embedding solution with dedicated protection that incorporates two techniques to decrease the network-resource consumption: (i) bandwidth squeezing, \ie the opportunity to tune the amount of bandwidth guaranteed in case of failures\cite{4528661, sone2011bandwidth, shen2014optimal, chen2015flexible, Assis:16}; and (ii) survivable multi-path provisioning~\cite{Huang:2011:SMP:2042972.2042976, khan2015simple}, \ie VLink demand splitting over multiple disjoint paths. Bandwidth squeezing is motivated by the fact that not all services require the full bandwidth to be operational, \eg a file transfer service can still be operational with reduced bandwidth during a failure. The extent of bandwidth reduction is measured by bandwidth squeezing rate (BSR), which is the percentage of original bandwidth made available after a failure. Survivable multi-path provisioning  allows us to exploit the fact that not all of the disjoint splits belonging to the same VLink are jointly affected by a single substrate link failure. Therefore, the bandwidth of the surviving paths after a failure can be reused to provide a given BSR. Note that, in this paper, we consider a single substrate link failure scenario, which is more common than multiple simultaneous link failures~\cite{studyisp}.

The illustrative example in Fig. \ref{fig:split-nodp}  demonstrates how the synergy of bandwidth squeezing and survivable multi-path provisioning can reduce backup resources for dedicated protection. If we do not allow any splitting of a VLink demand (\ie in our example, a 600Gbps demand is carried over a single primary path  $A$-$B$-$C$), then we require 2$\times$ the original bandwidth for providing a dedicated protection path ($A$-$E$-$C$), as illustrated in \fig{subfig:nosplit}. However, if we split the demand over three disjoint paths as in \fig{subfig:split}, then any two of the three paths will be active after a single substrate link failure. In this way 100\% bandwidth will still be available after a failure, even though we are allocating a total bandwidth that is only 1.5$\times$ the original demand. Bandwidth savings can be even higher if we employ bandwidth squeezing, \eg \fig{subfig:split-bsr} shows how 66\% BSR is ensured without any additional backup bandwidth. However, the advantages of multi-path provisioning depend on the availability of disjoint paths in the substrate. For instance, the lack of sufficient number of disjoint paths reduces the gain of multi-path provisioning regardless of the extent of BSR as shown in \fig{subfig:split-ndp-bsr} and  \fig{subfig:split-ndp}.

In this paper, we present our solutions for reliable slicing of 5G transport networks based on VN embedding with dedicated protection over an EON. 
Our solutions leverage EON's capabilities in finer-grained spectrum allocation, adapting modulation format and Forward Error Correction (FEC) overhead for rightsize resource allocation to the VNs. We also perform simulations using realistic network topologies, which provide valuable insight into how different levels of BSR and path diversity in the substrate EON impact the extent of backup resource savings for dedicated protection.

The rest of the paper is organized as follows. We first discuss the relevant research literature in \sect{sec:related}. Then, we formulate the optimal problem solution as an Integer Linear Program (ILP) in \sect{sec:ilp} followed by a heuristic algorithm to solve large problem instances in \sect{sec:heuristic}. 
In \sect{sec:eval}, we discuss the numerical evaluation of the proposed approaches over realistic network instances. Finally, we present our concluding remarks in \sect{sec:conclusion}.

%% file: related.tex
\section{Related Works}\label{sec:related}
Reliability aspects of VN embedding~\cite{rahman2010survivable} and, more generally, of traffic provisioning in fixed-grid optical networks~\cite{ramamurthy2003survivable} have been extensively studied. Techniques to enhance reliability~\cite{ramamurthy2003survivable} are typically categorized in two broad classes: protection and restoration. The first class contains proactive approaches where backup paths are provisioned at the time a VN or a traffic matrix is provisioned, which leads to faster recovery time. Backup paths can be either dedicated or shared among multiple requests.  In contrast, restoration techniques provision backup resources after a failure has occurred. Restoration alleviates the issue of keeping a large amount of backup resources idle during regular operation, however it has significantly higher recovery time, not suitable for supporting ultra-reliable communications in 5G networks~\cite{ramamurthy2003survivable, castro2012single}, and will not be considered in this study. Similarly, shared protection could be used to reduce the amount of backup resources of dedicated protection~\cite{liu2013survivable, shen2014optimal, 7506633, cai2016multicast, yin2016shared}, but at the cost of longer protections switching time~\cite{ou2006survivable, hauser2002capacity}, hence in our study we focus on dedicated protection.


The classical routing and spectrum allocation problem in EON with dedicated path protection has been addressed in~\cite{klinkowski2013genetic, klinkowski2013evolutionary, walkowiak2014routing}. However, the problem of reliable VNE over EON has gained attention only recently. Some research works have addressed the VN embedding problem over EON with dedicated path protection~\cite{xie2014survivable, chen2016cost, zhao2016energy} while minimizing spectrum  occupation, number of regenerators \cite{xie2014survivable, chen2016cost} or energy consumption~\cite{zhao2016energy}. These works formulated the problem using a path-based ILP, where a disjoint backup path for each of the usable primary paths is precomputed. Also, they assume the modulation format is already selected for each bitrate (\eg 100Gbps bitrate is provisioned using only DP-QPSK modulation format), further limiting the solution space. None of these works consider multi-path provisioning or capture all the tunable transmission parameters made available by EONs, as we do in this work. 


To save network resources, a number of research studies considered a weaker form of protection, called bandwidth squeezing, where only a reduced bandwidth is guaranteed  after a failure, according to a given BSR~\cite{4528661, sone2011bandwidth, shen2014optimal, chen2015flexible, Assis:16, goscien2016survivable}. The advantages of using bandwidth squeezing can be amplified when combined with virtual link demand splitting over multiple disjoint substrate paths~\cite{4528661, Assis:16, ruan2013survivable}. The latter is commonly known as survivable multi-path provisioning~\cite{Huang:2011:SMP:2042972.2042976, ruan2013survivable, goscien2016survivable, khan2015simple}. Majority of the survivable multi-path approaches, except~\cite{goscien2016survivable}, assume that all the paths used for embedding a virtual link are link disjoint, and may not explore the complete solution space. In contrast to these existing works, our work stands out by considering a VN topology embedding and by allowing splitting of a virtual link demand not only over non-disjoint paths but also over multiple spectral segments on the same path, significantly increasing the complexity of the problem as discussed in \sect{sec:ilp} and \ref{sec:heuristic}.  





%% file: model.tex
\section{Mathematical Model and Problem Statement}\label{sec:model}

\subsection{Substrate EON}\label{subsec:sn}
The substrate EON (SN) is an undirected graph $G = (V, E)$, where $V$ and $E$ are the set of substrate optical nodes (SNodes) and substrate optical links (SLinks), respectively. We assume the SLinks to be bi-directional, \ie adjacent optical nodes are connected by one optical fiber in each direction. The optical frequency spectrum on each SLink $e = (u, v) \in E$ is divided into equal-width frequency slices represented by the set $S$ and enumerated as $1, 2 $\ldots$ |S|$.  $\mathcal{P}$ and $\mathcal{P}^k_{uv} \subset \mathcal{P}$ represent the set of all paths in $G$ and the set of $k$-shortest paths between nodes $u, v \in V$, respectively. The number of SLinks and the physical length of a path $p$ in kilometers are represented by $|p|$ and $len(p)$, respectively. We use the binary variable $\delta_{pe}$ to denote the association between a path $p \in \mathcal{P}$ and an SLink $e \in E$.

The following transmission parameters can be configured on a path $p$ with length $len(p)$ to enable data transmission with different data-rates $d \in \mathcal{D}$: \textit{baud-rate} or \textit{symbol-rate}, $b$, \textit{modulation format}, $m$, and \textit{FEC} \textit{overhead}, $f$, selected from the set of possible values $\mathcal{B}$, $\mathcal{M}$, and $\mathcal{F}$, respectively. We use a tuple $t = (d, b, m, f) \in \mathcal{T} = (\mathcal{D} \times \mathcal{B} \times \mathcal{M} \times \mathcal{F})$ to represent a transmission configuration that dictates the combination of $b \in \mathcal{B}$, $m \in \mathcal{M}$, and $f \in \mathcal{F}$ yielding a data-rate $d \in \mathcal{D}$. For the sake of representation we use $t^{(d)}, t^{(b)}, t^{(m)}$, and $t^{(f)}$ to denote the data-rate, baud-rate, modulation format, and FEC overhead of a configuration $t \in \mathcal{T}$. A \textit{reach table} $\mathcal{R}$, computed based on physical layer characteristics, specifies the maximum length of a path (\ie the \textit{reach} $r_t$) capable of retaining a satisfactory optical signal to noise ratio when configured according to a transmission configuration $t \in \mathcal{T}$. Finally, $n_t$ denotes the number of slices required for a transmission configuration $t \in \mathcal{T}$, which is dependent on the parameters of $t$.
\subsection{Virtual Network}\label{subsec:vn}
The virtual network (VN) is represented by an undirected graph $\bar{G} = (\bar{V}, \bar{E})$, where $\bar{V}$ and $\bar{E}$ are the set of virtual nodes (VNodes) and virtual links (VLinks), respectively. The function $\tau: \bar{V} \rightarrow V$ represents VNode to SNode mapping and is an input to our problem (a common assumption for optical network virtualization~\cite{zhang2012survey}). Each virtual link $\bar{e} \in \bar{E}$ has a bandwidth requirement $\bar{\beta}_{\bar{e}}$ and a reliability requirement $0 < BSR_{\bar{e}} \leq 100$, which indicates the percentage of original bandwidth that should be available after an SLink fails. $BSR_{\bar{e}}$ can also be used to realize a bandwidth profile with a maximum and minimum demand similar to~\cite{hadi2019dynamic}. We allow VLinks to be mapped on multiple substrate paths (SPaths) (similar to~\cite{pages2014optimal, shahriarinfocom19}), each with a lower data-rate than $\bar{\beta}_{\bar{e}}$. Splitting $\bar{\beta}_{\bar{e}}$ over multiple SPaths is a feasible way to support higher data-rates (\eg $\ge400$Gbps) that limit the number of usable paths due to their shorter reaches. However, we restrict the number of VLink splits to maximum $q$ $(\geq 1)$. Such multi-path embedding is supported by technologies such as Virtual Concatenation (VCAT) in Optical Transport Network (OTN)~\cite{bernstein2006vcat} or bonding capabilities of FlexEthernet~\cite{flexe}. 

\subsection{Problem Statement}\label{subsec:ps}
Given an SN $G$, a reach table $\mathcal{R}$, and a VN request $\bar{G}$ with given VNode mapping function $\tau$.:
\begin{itemize}
	\item Compute the link embedding function $\gamma : \bar E \rightarrow \chi : \chi \subset \mathcal{P} \times \mathcal{T} \times S^2$ and $1 \leq |\chi| \leq q$, \ie compute up to a maximum of $q$ splits for each VLink $\bar{e} \in \bar{E}$ such that $0.01\times BSR_{\bar{e}}\times \bar{\beta}_{\bar{e}}$ bandwidth is available during an SLink failure and at least $\bar{\beta}_{\bar{e}}$ bandwidth is available during the rest of the time. For each split, $\gamma$ should select an SPath and an appropriate transmission configuration $t \in \mathcal{T}$ from the reach table $\mathcal{R}$, and allocate a contiguous segment of slices represented by the starting and ending slice index on each SLink along the SPath. Note that the same SPath can be used multiple times as the splits of a VLink following the reasoning in \cite{shahriarinfocom19}. $\chi_{\bar{e}i} = (p, t, s_b, s_t) |1 \leq i \leq q$ represents the $i$-th split, where $\chi_{\bar{e}i}^{(p)}$ and $\chi_{\bar{e}i}^{(t)}$ denote the selected SPath and transmission configuration for the $i$-th split, respectively. In addition, allocation of spectrum slices for the $i$-th split begins at index $\chi_{\bar{e}i}^{(s_b)}$ and ends at index $\chi_{\bar{e}i}^{(s_t)}$ along each SLink in the SPath $\chi_{\bar{e}i}^{(p)}$. 
	

	
	\item The total number of slices required to provision the VN is minimum according to the following cost function:\vspace{-0.25em}
	\begin{equation}
		\sum_{\forall \bar{e} \in \bar{E}} \sum_{i = 1}^{q} (\chi_{\bar{e}i}^{(s_t)} - \chi_{\bar{e}i}^{(s_b)} + 1) \times |\chi_{\bar{e}i}^{(p)}|
	\label{cost-eqn}\end{equation}
	Here, $|\chi_{\bar{e}i}^{(p)}|$ is the number of SLinks on the SPath $\chi_{\bar{e}i}^{(p)}$.
\end{itemize}
The above is subject to substrate resource constraints, and spectral contiguity (\ie the allocated slices of each split are always adjacent to each other) and continuity (\ie the same sequence of slices are allocated on each SLink along an SPath) constraints on the lightpaths.

\subsection{Pre-computations}\label{subsec:pre}
For each VLink $\bar{e} \in \bar{E}$, we pre-compute $\mathcal{P}_{\bar{e}}^k$, a set of $k$ shortest paths between the pair of SNodes where the VLink's endpoints' are mapped. For each SPath $p \in \mathcal{P}_{\bar{e}}^k$, we pre-compute the set of admissible transmission configurations, $\mathcal{T}_{\bar{e}p} \subset \mathcal{T}$, such that each configuration $t \in \mathcal{T}_{\bar{e}p}$ results in a reach $r_t \geq len(p)$ and has a data-rate $t^{(d)}$. $\mathcal{T}_{\bar{e}}$ contains all the distinct tuples suitable for $\bar{e}$ and is defined as $\bigcup_{\forall p \in \mathcal{P}_{\bar{e}}^k} \mathcal{T}_{\bar{e}p}$.



%% file: ilp-formulation.tex
\section{Problem Formulation}\label{sec:ilp}
We present a path-based ILP formulation for optimally solving our problem. Note that some of the constraints, except the reliability constraints, have been presented in different forms in different research works~\cite{chowdhury2009virtual, wang2011study, shahriarinfocom19}. In the interest of completeness, we report below all the constraints. 

\subsection{Decision Variables}\label{subsec:c_dvariables}

We allow a VLink's bandwidth demand to be satisfied by provisioning slices over one or more SPaths where an SPath can be used more than once (up to a maximum of $q$) as discussed in \ref{subsec:ps}. To model the same SPath appearing more than once in a VLink's embedding, we assume each transmission configuration on an SPath can be instantiated multiple times (up to a maximum of $q$ times). The following variable represents VLink mapping:
\begin{align*}
w_{\bar{e}pti} = \begin{cases}
1 & \text{if } \bar{e} \in \bar{E} \text{ uses } i\text{-th} \text{ instance of } t \in \mathcal{T}_{\bar{e}p}\\
  & \text{ on path } p \in \mathcal{P}_{\bar{e}}^k\\
0 & \text{otherwise}
\end{cases}
\end{align*}
Finally, the following decision variable creates the relationship between a mapped SPath and the slices in its SLinks:
\begin{align*}
y_{\bar{e}ptis} = \begin{cases}
1 & \text{if } \bar{e} \in \bar{E} \text{ uses slice } s \in S \text{ on path } p \in \mathcal{P}_{\bar{e}}^k\\
  & \text{with the } i\text{-th} \text{ instance of } t \in \mathcal{T}_{\bar{e}p}\\
0 & \text{otherwise}
\end{cases}
\end{align*}
\vspace{-1em}

\subsection{Constraints}\label{subsec:c_constraints}
\subsubsection{VLink demand constraints}
We provision a VLink by splitting it across multiple (up to $q$) SPaths. Constraint \eqn{eq:demand} ensures that for each VLink $\bar{e} \in \bar{E}$, the sum of data-rates resulting from applying the selected transmission configuration on the selected splits is equal or larger than the VLink's demand. \eqn{eq:max-split} enforces an upper limit on the number of splits.
\begin{align}
\forall \bar{e} \in \bar{E}: \bar{\beta}_{\bar{e}} \leq \displaystyle\sum_{\forall p \in \mathcal{P}_{\bar{e}}^k} \sum_{\forall t \in \mathcal{T}_{\bar{e}p}} \sum_{i = 1}^{q} (w_{\bar{e}pti}\times t^{(d)}) \label{eq:demand}
\end{align}
\begin{align}
\forall \bar{e} \in \bar{E}: \displaystyle\sum_{\forall p \in \mathcal{P}_{\bar{e}}^k} \sum_{\forall t \in \mathcal{T}_{\bar{e}p}} \sum_{i = 1}^{q} w_{\bar{e}pti} & \leq q \label{eq:max-split}
\end{align}

\subsubsection{Slice assignment and Spectral Contiguity constraints}
We ensure by \eqn{eq:sl-demand} that if a path $p$ is selected with a specific transmission configuration $t$, then the required number of slices $n_t$ to support the data-rate $t^{(d)}$ is allocated on the path. \eqn{eq:sl-once} ensures that each slice on an SLink is allocated to at most one path. Finally, \eqn{eq:contiguity} ensures the slices allocated on each link of a path form a contiguous frequency spectrum.
\begin{flalign}
\linespread{0.5}
& \forall \bar{e} \in \bar{E}, \forall p \in \mathcal{P}_{\bar{e}}^k, \forall t \in \mathcal{T}_{\bar{e}p}, 1 \leq i \leq q: \displaystyle\sum_{\forall s \in S} y_{\bar{e}ptis} = n_t  w_{\bar{e}pti}\label{eq:sl-demand}\vspace{-1.5em}
\end{flalign}
\begin{flalign}
& \forall e \in E, \forall s \in S: \displaystyle\sum_{\forall \bar{e} \in \bar{E}} \sum_{\forall p \in \mathcal{P}_{\bar{e}}^k} \sum_{\forall t \in \mathcal{T}_{\bar{e}p}} \sum_{i = 1}^{q} w_{\bar{e}pti} y_{\bar{e}ptis} \delta_{pe}  \leq 1 \label{eq:sl-once}\\\vspace{-1.5em}
& \forall \bar{e} \in \bar{E}, \forall p \in \mathcal{P}_{\bar{e}}^k, \forall t \in \mathcal{T}_{\bar{e}p}, 1 \leq i \leq q, 1 \leq s \leq |S| - 1:\nonumber\\ & \displaystyle \sum_{s' = s+2}^{|S|} y_{\bar{e}ptis'} \leq |S| \times (1-y_{\bar{e}ptis}+ y_{\bar{e}pti(s+1)})\label{eq:contiguity}\vspace{-.75em}
\end{flalign}


\subsubsection{Reliability constraints}
\eqn{eq:rel} ensures that for each single substrate link failure scenario, the aggregate data rate of the unaffected splits of a VLink $\bar{e} \in \bar{E}$ is at least $BSR_{\bar{e}}$ percentage of the original VLink demand.
\begin{flalign}
 \nonumber \forall \bar{e} \in \bar{E}, \forall e \in E: (\displaystyle \sum_{\forall p \in \mathcal{P}_{\bar{e}}^k} \sum_{\forall t \in \mathcal{T}_{\bar{e}p}} \sum_{i = 1}^{q} w_{\bar{e}pti}  t^{(d)}) \\
- (\sum_{\forall p \in \mathcal{P}_{\bar{e}}^k} \sum_{\forall t \in \mathcal{T}_{\bar{e}p}} \sum_{i = 1}^{q} w_{\bar{e}pti}  t^{(d)}  \delta_{pe})  \geq 0.01  BSR_{\bar{e}}  \bar{\beta}_{\bar{e}} \label{eq:rel}
\end{flalign}

\subsection{Objective Function}\label{subsec:c_objective}
Our cost function minimizes the total number of spectrum slices required to embed all the VLinks of a VN as shown in the first part of \eqn{obj}. However, to break ties among multiple solutions with the same total number of slices, we use the second term with a fractional weight $\epsilon$ in \eqn{obj} that minimizes the  number of splits over all the VLinks. This gives us the following objective function:\vspace{-.5em}
\begin{equation}
\linespread{0.5}
\begin{split}
\text{minimize} (\displaystyle\sum_{\forall \bar{e} \in \bar{E}} \sum_{\forall p \in \mathcal{P}_{\bar{e}}^k} \sum_{\forall t \in \mathcal{T}_{\bar{e}p}} \sum_{i=1}^{q} \sum_{\forall s \in S} y_{\bar{e}ptis} |p| + \\ \epsilon \sum_{\forall \bar{e} \in \bar{E}} \sum_{\forall p \in \mathcal{P}_{\bar{e}}^k} \sum_{\forall t \in \mathcal{T}_{\bar{e}p}} \sum_{i=1}^{q} w_{\bar{e}pti})\end{split}
\label{obj}
\end{equation}

%% file: heuristic.tex
\section{Heuristic Algorithm}\label{sec:heuristic}
Given the limited scalability of the ILP formulation, we develop a heuristic algorithm to solve large instances of our problem. In this section, we first give an overview of the main steps involved in reliable VN embedding for a given node mapping (\sect{sec:vn-heuristic}). Then, we discuss how we select the order of VLinks to be embedded for increasing the chances of finding a feasible solution (\sect{sec:heuristic-ordering}). Finally, we discuss the embedding process of a single VLink (\sect{sec:single-heuristic}).


\subsection{Heuristic Solution for Reliable VN Embedding}\label{sec:vn-heuristic}
\SetKwProg{Fn}{function}{}{}
{
    \vspace{-1em}
	\linespread{0.7}
	\small
	\begin{algorithm}[!ht]
		\DontPrintSemicolon
		\caption{Algorithm for VN Embedding}
		\label{alg:Compute-VN}
		\Fn(){VNEmbedding($G, \bar{G}$, $\tau$)}{
			$\mathcal{\bar{E}} \gets$ GetVLinkOrder($G, \bar{G}$)\\
			\ForEach{$\bar{e} \in \mathcal{\bar{E}}$ in the sorted order}{
				$\mathcal{I}_{\bar{e}} \gets $FindEmbedding($G, \bar{e}, \mathcal{P}_{\bar{e}}^i, \mathcal{T}_{\bar{e}}$)\\
					
					\ForEach{$e \in p| p \in \mathcal{I}_{\bar{e}}.\mathbb{P}$}
					{
						Perform slice assignment using $\mathcal{I}_{\bar{e}}.\mathcal{S}$
					}
					$\chi_{\bar{e}}.\mathcal{P} \gets \mathcal{I}_{\bar{e}}.\mathbb{P}$, $\chi_{\bar{e}}.\mathcal{T} \gets \mathcal{I}_{\bar{e}}.\mathbb{T}_{\mathbb{P}}$\\
				\If{$\chi_{\bar{e}} = \phi$}{\Return $<\phi, \phi>$}
			}
			\Return $<\tau: \bar{V} \rightarrow V, \gamma : \bar E \rightarrow \chi>$
			
		}
	\end{algorithm}
	\vspace{-0.75em}
}
\alg{alg:Compute-VN} takes as input a VN $\bar{G}$, an EON $G$, and a node mapping function $\tau: \bar{V} \rightarrow V$. One way of computing a cost efficient VLink mapping for a given VNode mapping is to consider all $|\bar{E}|!$ possible orders for sequentially embedding VLinks. The lowest cost mapping among all possible VLink orders then can be chosen as a cost effective solution for the given VNode mapping. However, this brute-force approach is not scalable. Instead, \alg{alg:Compute-VN} considers only one sequential VLink order that is computed to converge to a solution within a reasonable time. To increase the chances of finding a feasible VN embedding, \alg{alg:Compute-VN} invokes \alg{alg:VLink-Ordering} that finds a good order of the VLinks ($\mathcal{\bar{E}}$) to embed. For each VLink $\bar{e} \in \mathcal{\bar{E}}$ in the computed order, \alg{alg:Compute-VN} finds the embedding solution based on $\mathcal{P}_{\bar{e}}^i$ by invoking \alg{alg:Compute-all} (line 4). \alg{alg:Compute-VN} then allocates spectrum slices on all SLinks present in the solution and updates the VLink embedding $\chi_{\bar{e}}$ accordingly. If no such solution for $\bar{e}$ can be found, the embedding of the VN is rejected. 

\subsection{Compute VLink Ordering}\label{sec:heuristic-ordering}



\alg{alg:VLink-Ordering} finds an order of VLinks, $O$, that increases the chances of finding a feasible embedding for all the VLinks. Computing a feasible embedding of a VLink depends on the availability of contiguous slices in the candidate SPaths of that VLink. If the candidate SPaths $\mathcal{P}_{\bar{e}_i}^k$ of a VLink $\bar{e}_i$ have many SLinks that are common with the candidate SPaths of other VLinks that appear before $\bar{e}_i$ in a VLink order $o$, the slices of these SLinks (correspondingly, SPaths) may be exhausted or fragmented in such a way that the embedding of $\bar{e}_i$ becomes infeasible. Note that VLinks that come after $\bar{e}_i$ in $o$ do not have any impact on the embedding of $\bar{e}_i$ even though $\bar{e}_i$ and the VLinks after $\bar{e}_i$ have common SLinks in their candidate SPath sets. Hence, effective commonality of $\bar{e}_i$ depends on $o$ and is defined as follows:
\begin{equation}
w(\bar{e}_i)^o = \sum_{\bar{e}_j \text{ precedes } \bar{e}_i \text{ in } o | \bar{e}_i \in \bar{E} \text{, } \bar{e}_j \in \bar{E} \text{, } \bar{e}_i \ne \bar{e}_j} w(\bar{e}_i, \bar{e}_j)
\label{eq:eq-node-weight1}
\end{equation}
where, $w(\bar{e}_i, \bar{e}_j)$ is the commonality between two VLinks $\bar{e}_i$ and $\bar{e}_j$ irrespective of any order and is defined as follows:
\begin{equation}
w(\bar{e}_i, \bar{e}_j) = |\{(p_{\bar{e}_i}, p_{\bar{e}_j}) : (p_{\bar{e}_i}, p_{\bar{e}_j}) \in \mathcal{P}_{\bar{e}_i}^k \times \mathcal{P}_{\bar{e}_j}^k \wedge p_{\bar{e}_i} \cap p_{\bar{e}_j} \neq \phi\}|
\label{eq:eq-edge-weight}
\end{equation}
A high value of $w(\bar{e}_i, \bar{e}_j)$ indicates more SPath pairs from the candidate SPath sets of $\bar{e}_i$ and $\bar{e}_j$ have common SLinks. Using \eqn{eq:eq-node-weight1}, we define the commonality index of $o$ as follows:
\begin{equation}
w^o= \max_{\bar{e}_i \in \bar{E}} w(\bar{e}_i)^o
\label{eq:eq-link-weight}
\end{equation}

To increase the chances of finding a feasible embedding for all the VLinks, \alg{alg:VLink-Ordering} finds an order $O$ that minimizes $w^o$ for all possible orders of $\bar{E}$. To do so, \alg{alg:VLink-Ordering} first constructs an auxiliary graph $Z=(N, A)$ based on a VN $\bar{G}$ and EON $G$ (lines 2). The auxiliary graph $Z=(N, A)$ is a weighted undirected graph where $N$ is the set of nodes and $A$ is the set of edges. There is a one-to-one correspondence between a node $n_{\bar{e}_i} \in N$ and a VLink $\bar{e}_i \in \bar{E}$. Therefore, an order of the nodes in $N$ corresponds to an order of the VLinks in $\bar{E}$. We include a weighted edge $(n_{\bar{e}_i}, n_{\bar{e}_j}) \in A$ between two distinct nodes $n_{\bar{e}_i} \in N$ and $n_{\bar{e}_j} \in N$ with weight $w(\bar{e}_i, \bar{e}_j)$. Using \eqn{eq:eq-edge-weight}, we define the weight of a node $n_{\bar{e}_i} \in N$ as follows:
\begin{equation}
w(n_{\bar{e}_i}) = \sum_{n_{\bar{e}_j} | (n_{\bar{e}_i}, n_{\bar{e}_j}) \in A} w(\bar{e}_i, \bar{e}_j)
\label{eq:eq-node-weight}
\end{equation}
To compute a node order (or a corresponding VLink order) that minimizes $w^o$, \alg{alg:VLink-Ordering} iteratively computes $w(n_{\bar{e}_i})$ using \eqn{eq:eq-node-weight} for all the nodes in $N$. Then \alg{alg:VLink-Ordering} chooses the node $n_{\bar{e}_i}^{min}$ with the minimum weight $w(n_{\bar{e}_i})$ and inserts the corresponding VLink to the last empty spot in the VLink order $O$ (line 15). Afterwards, the algorithm updates $Z$ by removing $n_{\bar{e}_i}^{min}$ and all of its incident edges. The updated auxiliary graph $Z$ allows us to use \eqn{eq:eq-node-weight} to compute $w(n_{\bar{e}_i})$ that corresponds to $w(\bar{e}_i)^O$ as updated $Z$ now only have all the nodes (and adjacent edges) that precede $n_{\bar{e}_i}^{min}$ in the node order, or equivalently $\bar{e}_i^{min}$ in $O$. \alg{alg:VLink-Ordering} repeats this process until all the VLinks are added to $O$.

\vspace{-0.5em}
\begin{restatable}{theorem}{onlytheorem}
\alg{alg:VLink-Ordering} returns a VLink ordering with the minimum commonality index. (See appendix for proof)
\end{restatable}
\SetKwProg{Fn}{function}{}{}
{
    \vspace{-.75em}
	\linespread{0.7}
	\small
	\begin{algorithm}[h!]
		\DontPrintSemicolon
		\caption{Alg. for finding a VLink embedding order}
		\label{alg:VLink-Ordering}
		\Fn(){GetVLinkOrder($G, \bar{G}$)}{
			$Z = (N, A) \leftarrow$ AuxGraph($\bar{G}$),
			$O[1...n] \gets \phi$, $i \gets n$\\
			\While{$N \neq \phi$}{
				\lForEach{$n_{\bar{e}_i} \in N$}{
                    Compute $w(n_{\bar{e}_i})$ using \eqn{eq:eq-node-weight}
				}
                $n_{\bar{e}_i}^{min} \gets n_{\bar{e}_i}$ with minimum $w(n_{\bar{e}_i})$\\
				$O[i] \gets \bar{e}_i$ corresponding to $n_{\bar{e}_i}^{min}$\\
				$N \gets N - \{n_{\bar{e}_i}\}$\\
                \lForEach{$(n_{\bar{e}_i}^{min}, n_{\bar{e}_j}) \in A$}{
                $A \gets A \setminus \{(n_{\bar{e}_i}^{min}, n_{\bar{e}_j})\}$
                }
				$i \gets i - 1$\\
			}
			\Return $O$\\
		}
	\end{algorithm}
	\vspace{-2.25em}
}

\subsection{Compute Embedding of a Single VLink}\label{sec:single-heuristic}
The embedding of a VLink computed by \alg{alg:Compute-all} consists of a multi-set of SPaths where each SPath in the multi-set has an associated transmission configuration and spectrum slice allocation. Recall from \sect{sec:intro} that bandwidth requirement for dedicated protection against single SLink failure for a VLink depends on BSR and the number of disjoint SPaths used in the multi-set of SPaths of a solution. To exploit disjointness of the SPaths in the candidate SPath set $\mathbb{P}_{\bar{e}}^k$ of a VLink $\bar{e}$, \alg{alg:Compute-all} first computes disjoint SPath groups from the SPaths in $\mathbb{P}_{\bar{e}}^k$. We define a disjoint SPath group $H_{\bar{e}}$ from $\mathcal{P}_{\bar{e}}^k$ as follows:
\begin{equation}
\begin{split}
 H_{\bar{e}} = \{\delta\mathcal{P}_{\bar{e}}^k | \delta\mathcal{P}_{\bar{e}}^k \subseteq \mathcal{P}_{\bar{e}}^k \text{and }|\delta\mathcal{P}_{\bar{e}}^k| > 1 \text{ and } \\ \text{the SPaths in } \delta\mathcal{P}_{\bar{e}}^k \text{ are link disjoint}\} 
 \end{split}
\label{eq:dis-grp}
\end{equation}  
Note that two SPaths belonging to two different disjoint SPath groups $H_{\bar{e}}^i$ and $H_{\bar{e}}^j$ may share an SLink and an SPath can appear in multiple disjoint SPath groups. The set of all disjoint SPath groups for a VLink $\bar{e}$ is denoted by $\mathcal{H}_{\bar{e}}$. For instance, in \fig{subfig:nosplit}, $\mathcal{H}_{\bar{pq}} = \{\{AB-BC, AD-DC\}, \{AB-BC, AE-EC\}, \{AD-DC, AE-EC\}, \{AB-BC, AD-DC, AE-EC\}\}$, and in \fig{subfig:split-ndp-bsr}, $\mathcal{H}_{\bar{pq}} = \{\{AD-DC, AE-EC\}, \{AD-DB-BC, AE-EC\}\}$. 

\alg{alg:Compute-all} needs to enumerate all non-empty subsets of $\mathcal{H}_{\bar{e}}$ to find the optimal solution. Such enumeration is not scalable as the number of subsets of $\mathcal{H}_{\bar{e}}$ grows exponentially with the size of $\mathcal{H}_{\bar{e}}$. Hence, \alg{alg:Compute-all} applies a heuristic to obtain a smaller set of disjoint SPath groups $\mathbb{H}_{\bar{e}} \subseteq \mathcal{H}_{\bar{e}}$ that includes the most probable SPaths to be used as the splits of $\bar{e}$. The heuristic is motivated by the fact that longer SPaths allow only lower order modulation formats with a higher FEC that may require a large number of spectrum slices. In addition, longer SPaths often consist of more intermediate hops, thus increasing the total spectrum usage as per \eqn{obj}. To exclude such longer SPaths, our heuristic should construct $\mathbb{H}_{\bar{e}}$ with those disjoint SPath groups whose average total distances are small. However, doing so may bias the algorithm to include disjoint SPath groups with smaller number of SPaths for keeping the average distance low, which can reduce the benefits of splitting over multiple disjoint paths. \alg{alg:Compute-all} circumvents this issue by considering disjoint SPath groups of all sizes. From each set of disjoint SPath groups with size $i$, \alg{alg:Compute-all} selects the first $\sigma$ disjoint SPath groups ranked by the increasing average distance of the group (Line 6--9). The value of $\sigma$ is an input to \alg{alg:Compute-all} that can be used to keep the size of $\mathbb{H}_{\bar{e}}$ small.

After computing $\mathbb{H}_{\bar{e}}$, \alg{alg:Compute-all} enumerates all non-empty subsets of $\mathbb{H}_{\bar{e}}$ to assign data-rates to each disjoint SPath group in the subset such that each group provides dedicated protection as per BSR requirement for its data-rate. For each subset $\delta\mathbb{H}_{\bar{e}} \subseteq \mathbb{H}_{\bar{e}}$, \alg{alg:Compute-all} selects $|\delta\mathbb{H}_{\bar{e}}|$ data-rates such that the sum of these data-rates equals to demand $\bar{\beta}_{\bar{e}}$ (Line 10). These combinations of data-rates is represented by a multi-set $\mathbb{D}_{\delta\mathbb{H}_{\bar{e}}}= (\mathcal{D}, m_1)$, where $\mathcal{D}$ is the set of all data-rates, and $m_1:\mathcal{D} \rightarrow N$ is the number of times a data-rate in $\mathcal{D}$ appears in $\mathbb{D}_{\delta\mathbb{H}_{\bar{e}}}$. $\mathcal{M}(\mathbb{D}_{\delta\mathbb{H}_{\bar{e}}})$ represents all possible multi-sets of $\mathbb{D}_{\delta\mathbb{H}_{\bar{e}}}$.



Since none of $\delta\mathbb{H}_{\bar{e}}$ and $\mathbb{D}_{\delta\mathbb{H}_{\bar{e}}}$ are ordered sets, \alg{alg:Compute-all} needs to enumerate all permutations of data-rates from $\mathbb{D}_{\delta\mathbb{H}_{\bar{e}}}$ to assign data-rates in $\mathbb{D}_{\delta\mathbb{H}_{\bar{e}}}$ to $|\delta\mathbb{H}_{\bar{e}}|$ SPath groups. To do so, \alg{alg:Compute-all} generates all permutations of data-rates for each multi-set $\mathbb{D}_{\delta\mathbb{H}_{\bar{e}}} \in \mathcal{M}(\mathbb{D}_{\delta\mathbb{H}_{\bar{e}}})$, denoted by $\zeta(\mathbb{D}_{\delta\mathbb{H}_{\bar{e}}})$ (Line 12). For each of these permutations of data-rates $\mathrm{D}_{\delta\mathbb{H}_{\bar{e}}} \in \zeta(\mathbb{D}_{\delta\mathbb{H}_{\bar{e}}})$, \alg{alg:Compute-all} assigns a data-rate $d_{H_{\bar{e}}}$ from $\mathrm{D}_{\delta\mathbb{H}_{\bar{e}}}$ to a disjoint SPath group $H_{\bar{e}} \in \delta\mathbb{H}_{\bar{e}}$ in the same order (Line 15). For each disjoint SPath group $H_{\bar{e}} \in \delta\mathbb{H}_{\bar{e}}$ with its assigned data-rate $d_{H_{\bar{e}}}$, the algorithm computes the data-rate $d_{p_{H_{\bar{e}}}}$ of an SPath $p \in H_{\bar{e}}$ using~\eqn{eq:heu-eq7} that takes into account the number of disjoint SPaths in $H_{\bar{e}}$ and $BSR_{\bar{e}}$ of the VLink (Line 16). $d_{p_{H_{\bar{e}}}}$ is computed in a way such that the SPaths in $H_{\bar{e}}$ protect $BSR_{\bar{e}}$ fraction of $d_{H_{\bar{e}}}$ under any single SLink failure. However, when $BSR_{\bar{e}}$ is low, protection requirement is not enough to provide the assigned data-rate $d_{H_{\bar{e}}}$. In this case, data-rate $d_{p_{H_{\bar{e}}}}$ for an SPath is computed by equally dividing $d_{H_{\bar{e}}}$ along each disjoint SPath in $H_{\bar{e}}$. Note that each SPath in $H_{\bar{e}}$ will have an equal data-rate $d_{p_{H_{\bar{e}}}}$.
\begin{align}
d_{p_{H_{\bar{e}}}} = max (\frac{d_{H_{\bar{e}}}\times BSR_{\bar{e}}}{100\times(|H_{\bar{e}}|-1)}, \frac{d_{H_{\bar{e}}}}{|H_{\bar{e}}|})
\label{eq:heu-eq7}
\end{align}
Since an SPath $p$ can appear in multiple disjoint SPath groups in $\delta\mathbb{H}_{\bar{e}}$, a consolidation step is introduced to compute the total data-rate assigned to $p$ (denoted by $d_{p_{\delta\mathbb{H}_{\bar{e}}}}$) using~\eqn{eq:heu-eq8} (Line 17). Since the possible data-rates in the reach table are discrete values, the ceiling function returns the nearest rounded up data-rate after the summation in~\eqn{eq:heu-eq8}.
\begin{align}
d_{p_{\delta\mathbb{H}_{\bar{e}}}} = \ceil [\bigg]{\sum_{\forall H_{\bar{e}} \in \delta\mathbb{H}_{\bar{e}}} d_{p_{H_{\bar{e}}}}}
\label{eq:heu-eq8}
\end{align}

After the consolidation step, we get a set of distinct (not necessarily disjoint) SPaths $P_{\delta\mathbb{H}_{\bar{e}}}$ and their assigned data-rates for a particular $\delta\mathbb{H}_{\bar{e}}$ and $\mathbb{D}_{\delta\mathbb{H}_{\bar{e}}}$. However, an SPath $p_{\delta\mathbb{H}_{\bar{e}}} \in P_{\delta\mathbb{H}_{\bar{e}}}$ may split its assigned data-rate $d_{p_{\delta\mathbb{H}_{\bar{e}}}}$ into smaller data-rates either to ensure spectrum contiguity or to minimize the number of slices. To enumerate these possibilities, \alg{alg:Compute-all} generates the set of all possible multi-sets $\mathcal{M}({\mathbb{P}_{\delta\mathbb{H}_{\bar{e}}}})$ out of the set $P_{\delta\mathbb{H}_{\bar{e}}}$ (Line 18). A multi-set $\mathbb{P}_{\delta\mathbb{H}_{\bar{e}}} \in \mathcal{M}({\mathbb{P}_{\delta\mathbb{H}_{\bar{e}}}})$ 
is defined as $\mathbb{P}_{\delta\mathbb{H}_{\bar{e}}} = (P_{\delta\mathbb{H}_{\bar{e}}}, m_2)$, where, $m_2:P_{\delta\mathbb{H}_{\bar{e}}} \rightarrow N$ is the number of times an SPath in $P_{\delta\mathbb{H}_{\bar{e}}}$ appears in $\mathbb{P}_{\delta\mathbb{H}_{\bar{e}}}$. For each multi-set $\mathbb{P}_{\delta\mathbb{H}_{\bar{e}}}$, \alg{alg:Compute-all} assigns data-rates to each instance of SPath in $\mathbb{P}_{\delta\mathbb{H}_{\bar{e}}}$. This is trivial for an SPath that appears once in $\mathbb{P}_{\delta\mathbb{H}_{\bar{e}}}$ ($m_2(p)=1$). For an SPath with $m_2(p)>1$, \alg{alg:Compute-all} distributes the assigned data-rate $d_{p_{\delta\mathbb{H}_{\bar{e}}}}$ into $m_2(p)$ splits by generating permutations of multi-sets of data-rates of size $m_2(p)$ (Line 21--23). To get the set of all data-rate permutations $\mathcal{M}(\mathrm{D}_{\mathbb{P}_{\delta\mathbb{H}_{\bar{e}}}})$ for $\mathbb{P}_{\delta\mathbb{H}_{\bar{e}}}$, \alg{alg:Compute-all} considers all possible ways to combine the data-rate permutations of each distinct SPath in $P_{\delta\mathbb{H}_{\bar{e}}}$ (Line 24).



Once we have an SPath multi-set $\mathbb{P}_{\delta\mathbb{H}_{\bar{e}}}$ and its data-rate permutation $\mathrm{D}_{\mathbb{P}_{\delta\mathbb{H}_{\bar{e}}}}$, \alg{alg:Compute-all} invokes MDP procedure to find feasible transmission configuration and spectrum slice allocation with the least slice requirement (Line 26). MDP procedure is adapted from~\cite{shahriarinfocom19} that first selects a transmission configuration to achieve a data-rate in $\mathrm{D}_{\mathbb{P}_{\delta\mathbb{H}_{\bar{e}}}}$ along an SPath in $\mathbb{P}_{\delta\mathbb{H}_{\bar{e}}}$ and allocates the slices required by the transmission configuration using First-fit~\cite{chatterjee2015routing}. Among all the feasible solutions returned by MDP procedure, \alg{alg:Compute-all} selects the one that minimizes \eqn{obj}.

\subsubsection{Running Time Analysis}\label{sec:vn-heuristic-time}
\alg{alg:Compute-all} explores all subsets of $\mathbb{H}_{\bar{e}}$ yielding $2^{|\mathbb{H}_{\bar{e}}|}$ possibilities. For each subset $\delta\mathbb{H}_{\bar{e}} \subseteq \mathbb{H}_{\bar{e}}$, \alg{alg:Compute-all} explores $\binom{|\mathcal{D}|+|\delta\mathbb{H}_{\bar{e}}|-1}{|\delta\mathbb{H}_{\bar{e}}|}$ data-rate multi-sets. The number of permutations of a multiset $\mathbb{D}_{\delta\mathbb{H}_{\bar{e}}}$ of cardinality $|\delta\mathbb{H}_{\bar{e}}|$ is given by $\frac{|\delta\mathbb{H}_{\bar{e}}|!}{\Pi_{d_j \in \mathcal{D}} m_{1}(d_j)!}$\cite{brualdi1992introductory}. This results in $\frac{(|\mathcal{D}|+|\delta\mathbb{H}_{\bar{e}}|-1)!}{(|\mathcal{D}|-1)! \times \Pi_{d_j \in \mathcal{D}} m_{1}(d_j)!}$ enumerations. Then \alg{alg:Compute-all} enumerates $\mathcal{M}({\mathbb{P}_{\delta\mathbb{H}_{\bar{e}}}})$ multi-sets of SPaths based on an SPath set $P_{\delta\mathbb{H}_{\bar{e}}}$, and this enumeration has an upper bound $\binom{q-1}{\frac{q}{2}} = \frac{\Pi_{j=\ceil[]{\frac{q}{2}}+1}^{q-1} j!}{\ceil[]{\frac{q}{2}}!}$. 
Assuming that for each SPath $p \in P_{\delta\mathbb{H}_{\bar{e}}}$, \alg{alg:Compute-all} can have a data-rate multi-set $\mathbb{D}_{p} = (\mathcal{D}, m_3)$, where $m_3:\mathcal{D} \rightarrow N$ is the frequency of a data-rate in $\mathbb{D}_{p}$, the number of data-rate permutations in $\mathcal{M}(\mathrm{D}_{\mathbb{P}_{\delta\mathbb{H}_{\bar{e}}}})$ is ${\Pi_{i=1}^{|P_{\delta\mathbb{H}_{\bar{e}}}|}}\frac{(|\mathcal{D}|+m_2(p_i)-1)!}{(|\mathcal{D}|-1)! \times \Pi_{d_j \in \mathcal{D}} m_{3}(d_j)!}$. Since MDP's time complexity is $\frac{q!}{\Pi_{p_j \in {\mathbb{P}_{\delta\mathbb{H}_{\bar{e}}}}} m_{2}(p_j)!}$ as per~\cite{shahriarinfocom19}, the running time of \alg{alg:Compute-all} becomes $2^{|\mathbb{H}_{\bar{e}}|} \times \frac{(|\mathcal{D}|+|\delta\mathbb{H}_{\bar{e}}|-1)!}{(|\mathcal{D}|-1)! \times \Pi_{d_j \in \mathcal{D}} m_{1}(d_j)!} \times \frac{\Pi_{j=\ceil[]{\frac{q}{2}}+1}^{q-1} j!}{\ceil[]{\frac{q}{2}}!} \times {\Pi_{i=1}^{|P_{\delta\mathbb{H}_{\bar{e}}}|}}\frac{(|\mathcal{D}|+m_2(p_i)-1)!}{(|\mathcal{D}|-1)! \times \Pi_{d_j \in \mathcal{D}} m_{3}(d_j)!} \times \frac{q!}{\Pi_{p_j \in \mathbb{P}_{\delta\mathbb{H}_{\bar{e}}}} m_{2}(p_j)!}$. As \alg{alg:Compute-all} keeps the size of $\mathbb{H}_{\bar{e}}$ small, the running time is dominated by the latter part. However, typical values of $|\mathcal{D}|$ and $q$ are small and we apply several pruning techniques to improve the running time.\vspace{-1em}

\SetKwProg{Fn}{function}{}{}
{
	\linespread{0.75}
	\small
	\begin{algorithm}[tbh!]
		\DontPrintSemicolon
		\caption{Find the embedding of a single VLink}
		\label{alg:Compute-all}
		\Fn(){FindEmbedding($G, \bar{e}, \mathcal{P}_{\bar{e}}^k, \mathcal{T}_{\bar{e}}, \sigma$)}{
            \While{A new $H_{\bar{e}}$ exists}
            {
                Compute a new disjoint path group $H_{\bar{e}}$ using \eqn{eq:dis-grp}\\
                $\mathcal{H}_{\bar{e}} \gets \mathcal{H}_{\bar{e}} \cup H_{\bar{e}}$\\
            }
            \For{$i = 2$ to $\max_{\forall H_{\bar{e}} \in \mathcal{H}_{\bar{e}}} |H_{\bar{e}}|$}{
                $Z_{\bar{e}}^i \gets \{H_{\bar{e}}| H_{\bar{e}} \in \mathcal{H}_{\bar{e}} \wedge |H_{\bar{e}}| = i\}$\\
                $\mathcal{Z}_{\bar{e}}^i \gets$ First $\sigma$\text{ }$H_{\bar{e}} \in Z_{\bar{e}}^i$ with smallest avg. dist. \\
                $\mathbb{H}_{\bar{e}} \gets \mathbb{H}_{\bar{e}} \cup \mathcal{Z}_{\bar{e}}^i$\\
            }
			\ForEach{$\delta\mathbb{H}_{\bar{e}} \subseteq \mathbb{H}_{\bar{e}}$}{
				$\mathcal{M}(\mathbb{D}_{\delta\mathbb{H}_{\bar{e}}}) \gets$ All-Multi-Set($\mathcal{D}, |\delta\mathbb{H}_{\bar{e}}|$) s.t. $\sum_{d \in \mathbb{D}_{\delta\mathbb{H}_{\bar{e}}}} d =  \bar{\beta}_{\bar{e}}$\\
				\ForEach{$\mathbb{D}_{\delta\mathbb{H}_{\bar{e}}} \in \mathcal{M}(\mathbb{D}_{\delta\mathbb{H}_{\bar{e}}})$}{
					$\zeta(\mathbb{D}_{\delta\mathbb{H}_{\bar{e}}}) \gets$ All-Permutation($\mathbb{D}_{\delta\mathbb{H}_{\bar{e}}}$)\\
					\ForEach{$\mathrm{D}_{\delta\mathbb{H}_{\bar{e}}} \in \zeta(\mathbb{D}_{\delta\mathbb{H}_{\bar{e}}})$}
					{
                        \ForEach{$H_{\bar{e}} \in \delta\mathbb{H}_{\bar{e}}$}{
                        $d_{H_{\bar{e}}} \gets \mathrm{D}_{\delta\mathbb{H}_{\bar{e}}}$[index of $H_{\bar{e}} \in \delta\mathbb{H}_{\bar{e}}$]\\
                        Compute $d_{p_{H_{\bar{e}}}}$ using \eqn{eq:heu-eq7} \\
                        }
                        Compute $d_{p_{\delta\mathbb{H}_{\bar{e}}}}$ using \eqn{eq:heu-eq8}\\
                        $\mathcal{M}({\mathbb{P}_{\delta\mathbb{H}_{\bar{e}}}}) \gets$ All-Multi-Set($P_{\delta\mathbb{H}_{\bar{e}}}, q$)\\
                        \ForEach{$\mathbb{P}_{\delta\mathbb{H}_{\bar{e}}} \in \mathcal{M}({\mathbb{P}_{\delta\mathbb{H}_{\bar{e}}}})$}
                        {
                        \ForEach{$p \in \mathbb{P}_{\delta\mathbb{H}_{\bar{e}}}$}
                            {
            				$\mathcal{M}(\mathbb{D}_{p}) \gets$ All-Multi-Set($\mathcal{D}, m_2(p)$) s.t. $\sum_{d \in \mathbb{D}_{p}} d =  d_{p_{\delta\mathbb{H}_{\bar{e}}}}$\\
            				\ForEach{$\mathbb{D}_{p} \in \mathcal{M}(\mathbb{D}_{p})$}{
            					$\zeta(\mathbb{D}_{p}) \gets$ All-Permutation($\mathbb{D}_{p}$)\\
                            }
                           } $\mathcal{M}(\mathrm{D}_{\mathbb{P}_{\delta\mathbb{H}_{\bar{e}}}}) \gets \zeta(\mathbb{D}_{p_1}) \times \zeta(\mathbb{D}_{p_2}) ... \times \zeta(\mathbb{D}_{p_{|P_{\delta\mathbb{H}_{\bar{e}}}|}})$
						  
						\ForEach{$\mathrm{D}_{\mathbb{P}_{\delta\mathbb{H}_{\bar{e}}}} \in \mathcal{M}(\mathrm{D}_{\mathbb{P}_{\delta\mathbb{H}_{\bar{e}}}})$}{
						$<n, \mathbb{P}, \mathbb{T}, \mathbb{S}> \gets$ MDP($\mathbb{P}_{\delta\mathbb{H}_{\bar{e}}}, \mathrm{D}_{\mathbb{P}_{\delta\mathbb{H}_{\bar{e}}}}$)
						}
                            }
                           }
                    }
                    }

		 Find $\mathbb{P}^{opt}$, $\mathbb{T}^{opt}$ and $\mathbb{S}^{opt}$ that minimizes \eqn{obj}\\
		$\mathcal{I}_{\bar{e}} \gets <\mathbb{P}^{opt}, \mathbb{T}^{opt}, \mathbb{S}^{opt}>$\\
			\Return $\mathcal{I}_{\bar{e}}$
		}
	\end{algorithm}
	\vspace{-0.5em}
}

%% file: evaluation.tex
\vspace{-1em}
\section{Evaluation}\label{sec:eval}\vspace{-1em}




\subsection{Simulation Setup} \label{subsec:expsetup}
We implement the ILP formulation from \sect{sec:ilp} using IBM ILOG CPLEX C++ libraries. We compare the ILP's solution with a C++ implementation of the heuristic presented in \sect{sec:heuristic}. Simulations are run on a machine with 8$\times$10-core Intel Xeon E7-8870 2.40GHz processors and $1$TB RAM. 

\textbf{SN characteristic:} We use Nobel Germany (17 nodes and 26 links) and Germany50 (50 nodes and 88 links) networks from SNDlib repository~\cite{SNDlib} as the SNs for small and large scale simulations, respectively. $k=25$ and $k=20$ shortest paths between all pairs of SNodes are pre-computed as inputs to our simulation for Nobel Germany and Germany50, respectively. Each SLink in the EON has spectrum bandwidth of 600GHz and 4THz for small and large scale simulations, respectively. We also vary the Link-to-Node Ratio (LNR) of Nobel Germany SN by adding or removing some links in the original topology (original LNR is 1.53), to show the impact of different level of SN densities on the performance metrics.

\textbf{VN generation:} We synthetically generate VNs with different properties. For small-scale scenario, we restrict VN sizes to 4 VNodes and 5 VLinks to limit the ILP's complexity. For large-scale, we generate VNs with 20 VNodes and 30 VLinks. VLink demands in both small-scale and large-scale scenarios are varied between 100Gbps to 1Tbps with possible values as multiples of 100Gbps. Node mapping of VNodes in SNodes is randomly chosen. In our simulations, we vary $BSR$ from $0\%$ to $100\%$, where $BSR=0$ means no protection and $BSR=100\%$ provides full dedicated protection. We generate $5$ and $20$ different VNs with similar total bandwidth demands for small- and large-scale simulations, respectively, and take the mean of the performance metrics over those VNs. 
\vspace{-1.5em}
\subsection{Compared Variants} \label{subsec:variants}
In our evaluation, we quantify the impact of the two key flexibilities introduced by EONs, namely adaptability in transponder configuration (i.e., variable FEC and modulation) and in spectrum allocation. To do so, we compare the variants listed in Table \ref{tab:comparison-summary}. \textit{Fix-RT} considers fixed-grid spectrum allocation with only one tuple for modulation format, FEC overhead, and reach for each fixed-grid data rate in $\mathcal{R}$ and serves as the baseline for our evaluation. In contrast, the other two  variants, \textit{Fix-AT} and \textit{Flex-AT}, exploit a transponder's capability to choose from a variety of modulation formats and/or FEC overheads for each data rate in $\mathcal{R}$, leading to different reaches. These two variants differ in terms of spectrum slice allocation granularity (\ie 50GHz and 12.5GHz for fixed-grid and flex-grid, respectively) and the total number of possible data rates (\ie 100G, 200G, 400G for fixed-grid and 100G, 150G, 200G, 250G, 300G, 400G, 500G, 600G, and 800G for flex-grid). Spectrum occupation and data rates of each variant are chosen based on current industry standards~\cite{itu-t}. We fix the value of maximum number of splits to $q=8$ for all cases.
\begin{table}  [tbh!]
	\centering
	\vspace{-1.5em}
	\caption{Compared Variants}
	\vspace{0em}
	\begin{tabular} {|m{3.0cm}|m{1.8cm}|m{2.6cm}|}\hline
    {\textbf{Variant}} &\textbf{Transponder Flexibility} & \textbf{Flexible Spectrum Allocation} \\ \hline
    \hline
	Fix-RT & No & No \\ \hline
    Fix-AT & Yes & No\\ \hline
	Flex-AT & Yes & Yes \\ \hline	

	\end{tabular}
	\label{tab:comparison-summary}
	\vspace{-1em}
\end{table}
\subsection{Performance Metrics}\label{subsec:performance-metrics}


\paragraph{Spectrum Slice Usage (SSU)} The total number of spectrum slices required to embed a VN. This metric is computed using only the first term of Eqn. \eqn{obj}.

\paragraph{Protection overhead} Ratio between total bandwidth allocated  of a VLink (on all different splits) and the actual bandwidth demand of the VLink. It is a measure of the required extra bandwidth for providing protection.

\paragraph{Max. no. of disjoint paths} Average of the maximum number of disjoint SPaths used to embed a VLink.

\paragraph{Max. no. of splits} Average of the maximum number of splits used to embed a VLink.



\begin{figure*}[!h]
	\centering
	\subfigure[Impact on SSU]{\includegraphics[trim=1cm 1cm 1.15cm 1.5cm,width=0.245\textwidth]{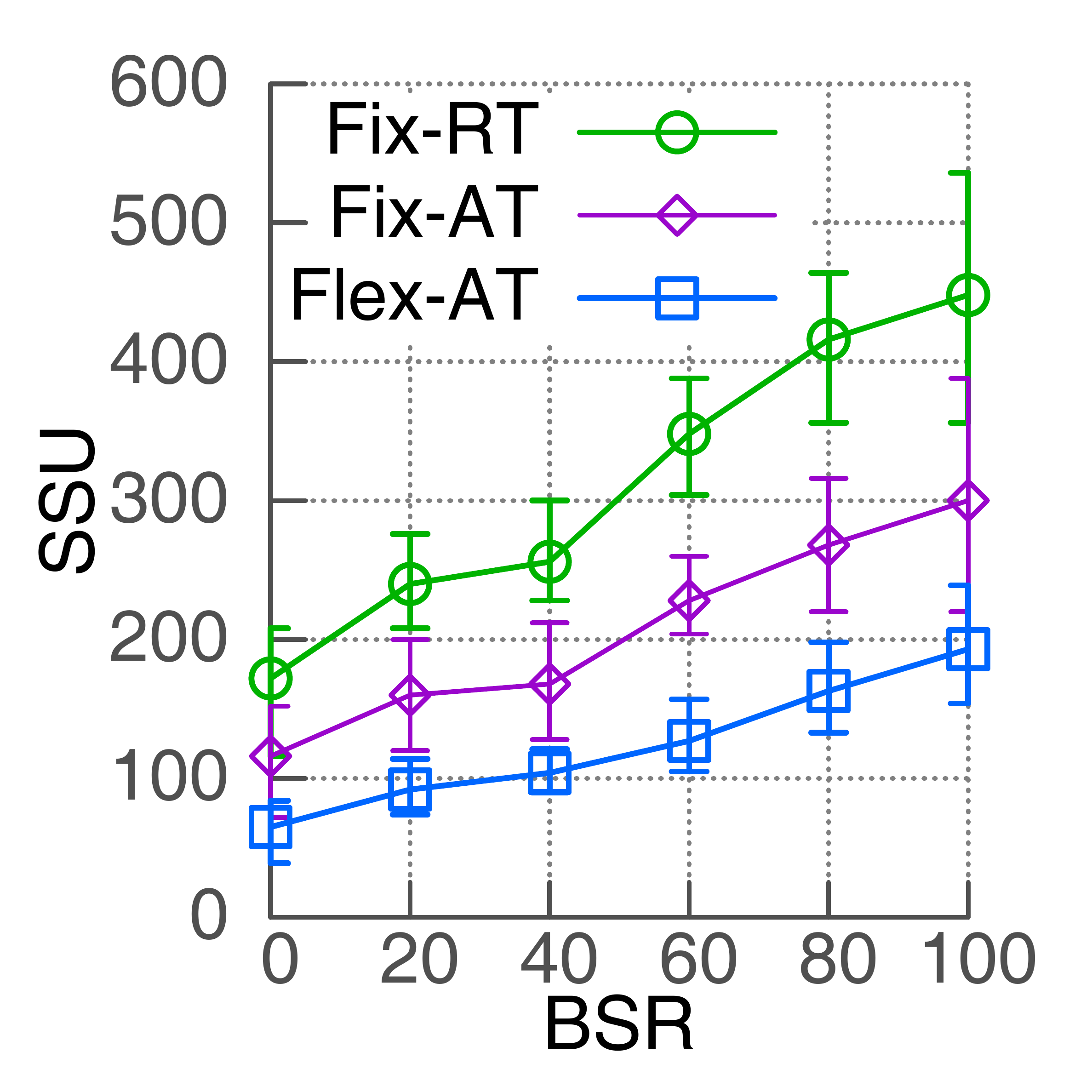}\label{fig_c:acc}}
	\subfigure[Impact on protection overhead]{\includegraphics[trim=1cm 1cm 1.15cm 1.5cm,width=0.245\textwidth]{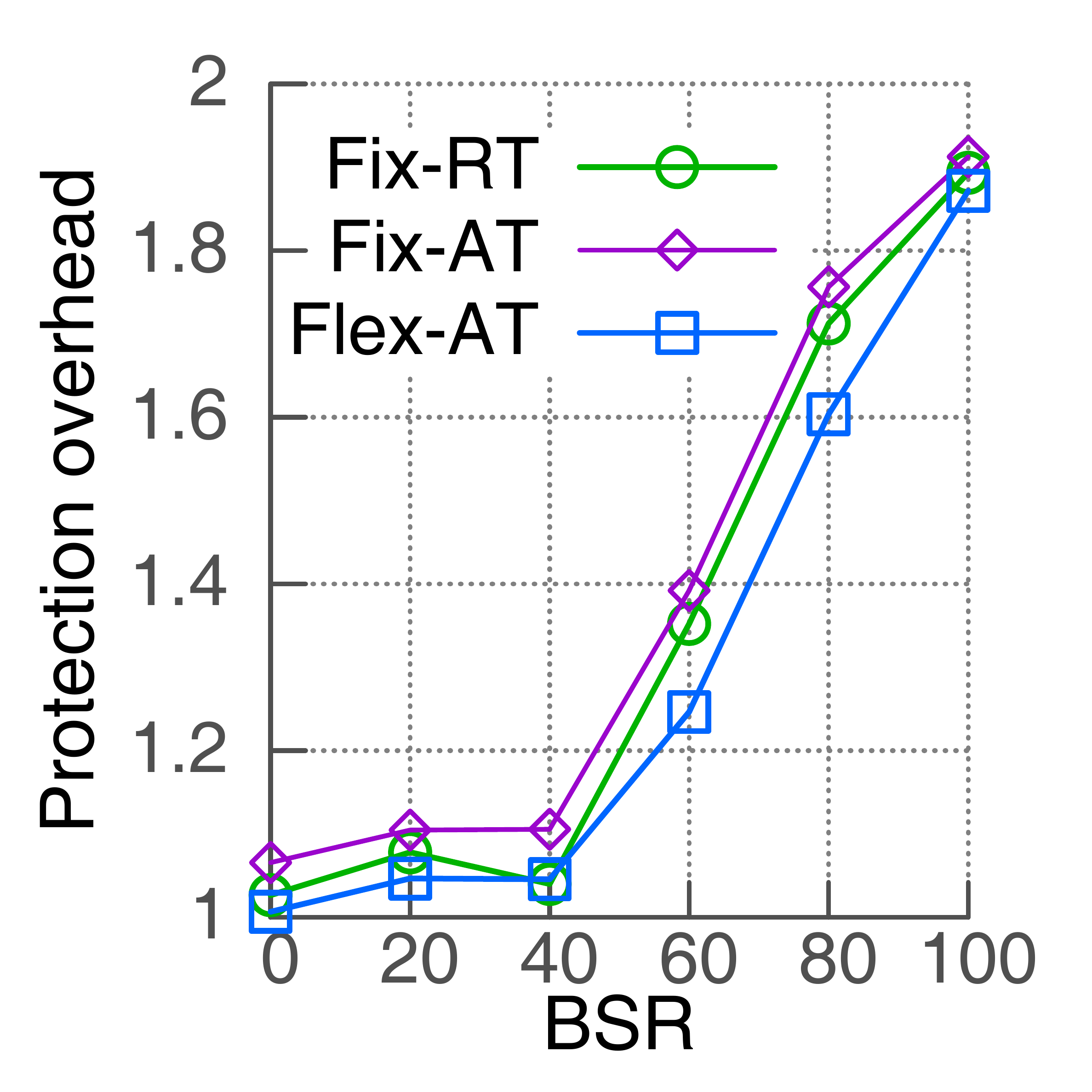}\label{fig_c:bacc}}
    \subfigure[Impact on number of disjoint paths]{\includegraphics[trim=1cm 1cm 1.15cm 1.5cm,width=0.245\textwidth]{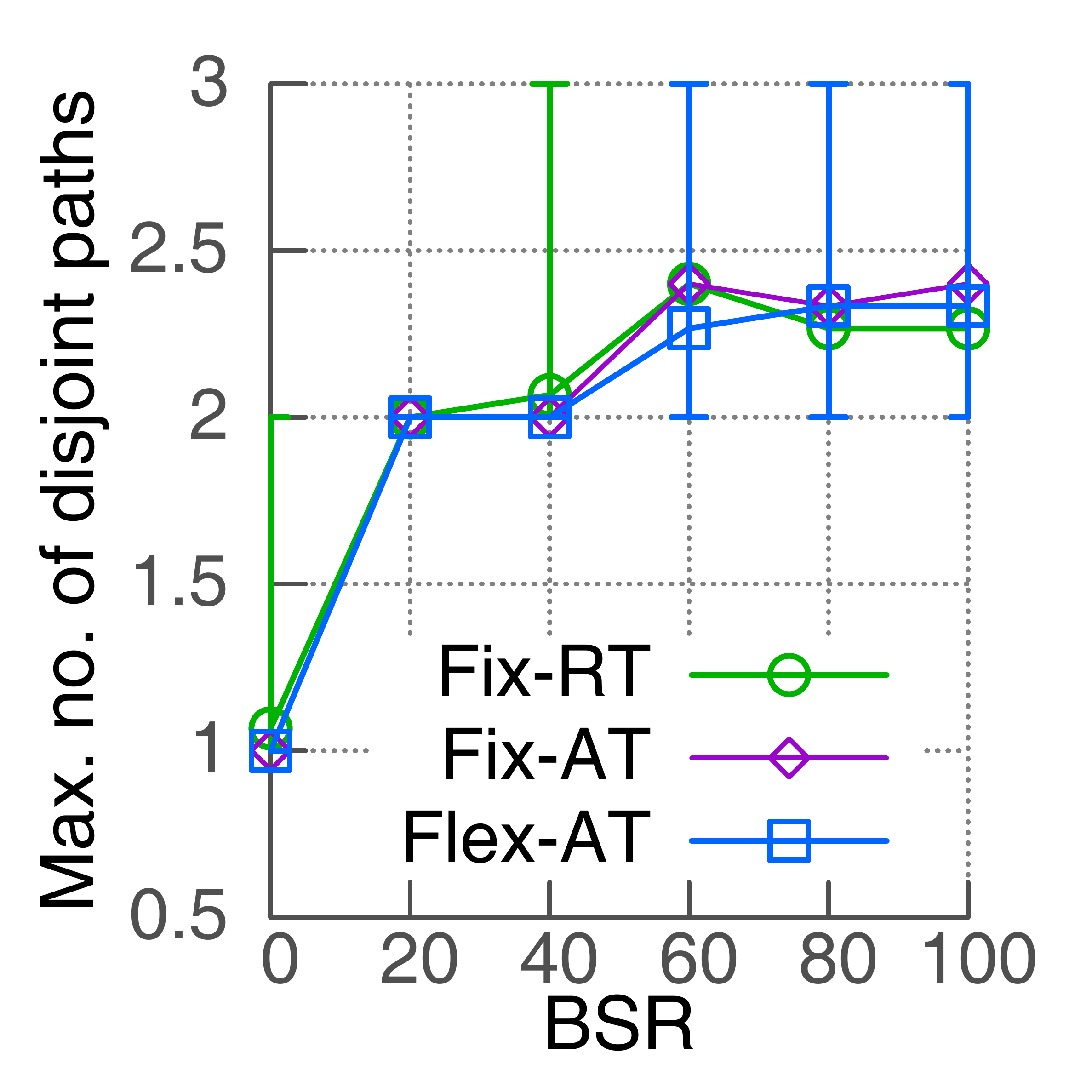}\label{fig_c:dist}}
    \subfigure[Impact on number of splits]{\includegraphics[trim=1cm 1cm 1.15cm 1.5cm,width=0.245\textwidth]{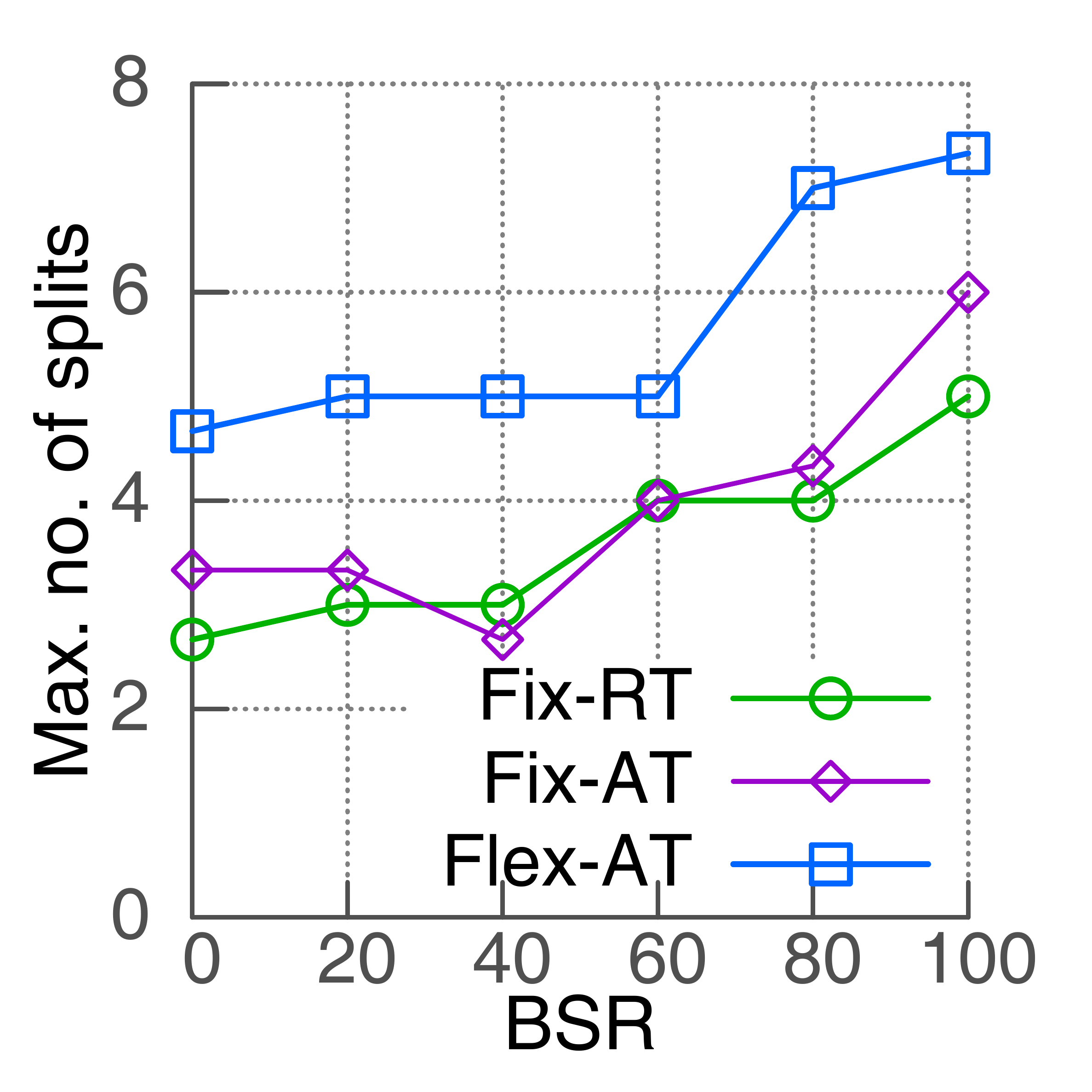}\label{fig_c:split}}
    \vspace{-1.0em}
	\caption{Impact of varying BSR on performance metrics in small Germany SN}
	\vspace{-1.0em}
	\label{fig_c:acratio}
\end{figure*}

\begin{figure}[!h]
	\centering
	\subfigure[Impact of SN density]{\includegraphics[trim=0cm 0.8cm 0.0cm 1.5cm,width=0.24\textwidth]{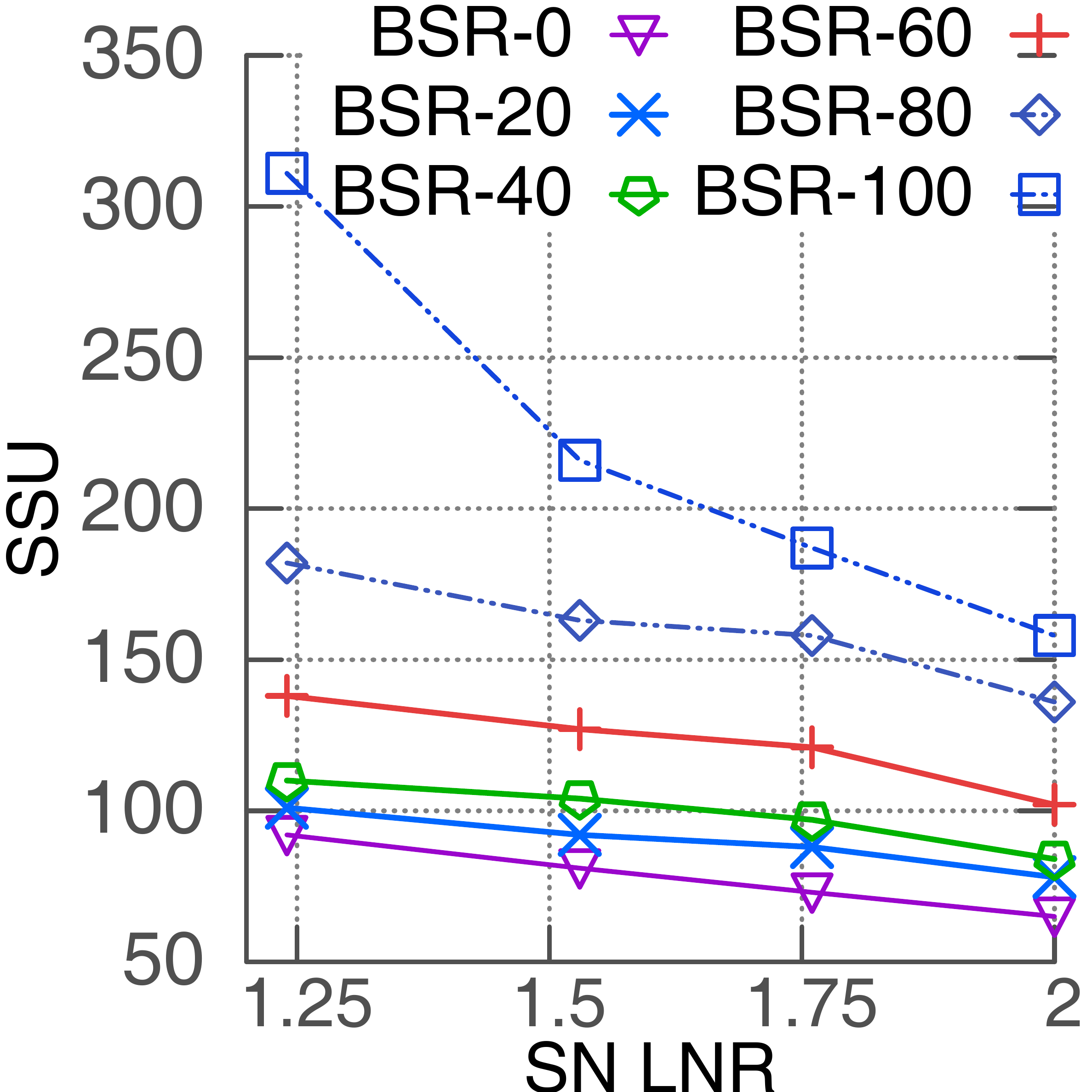}\label{fig_c:bacc_sn}}
	\subfigure[Split vs noSplit on same path]{\includegraphics[trim=0cm 0.8cm 0.5cm 1.5cm,width=0.24\textwidth]{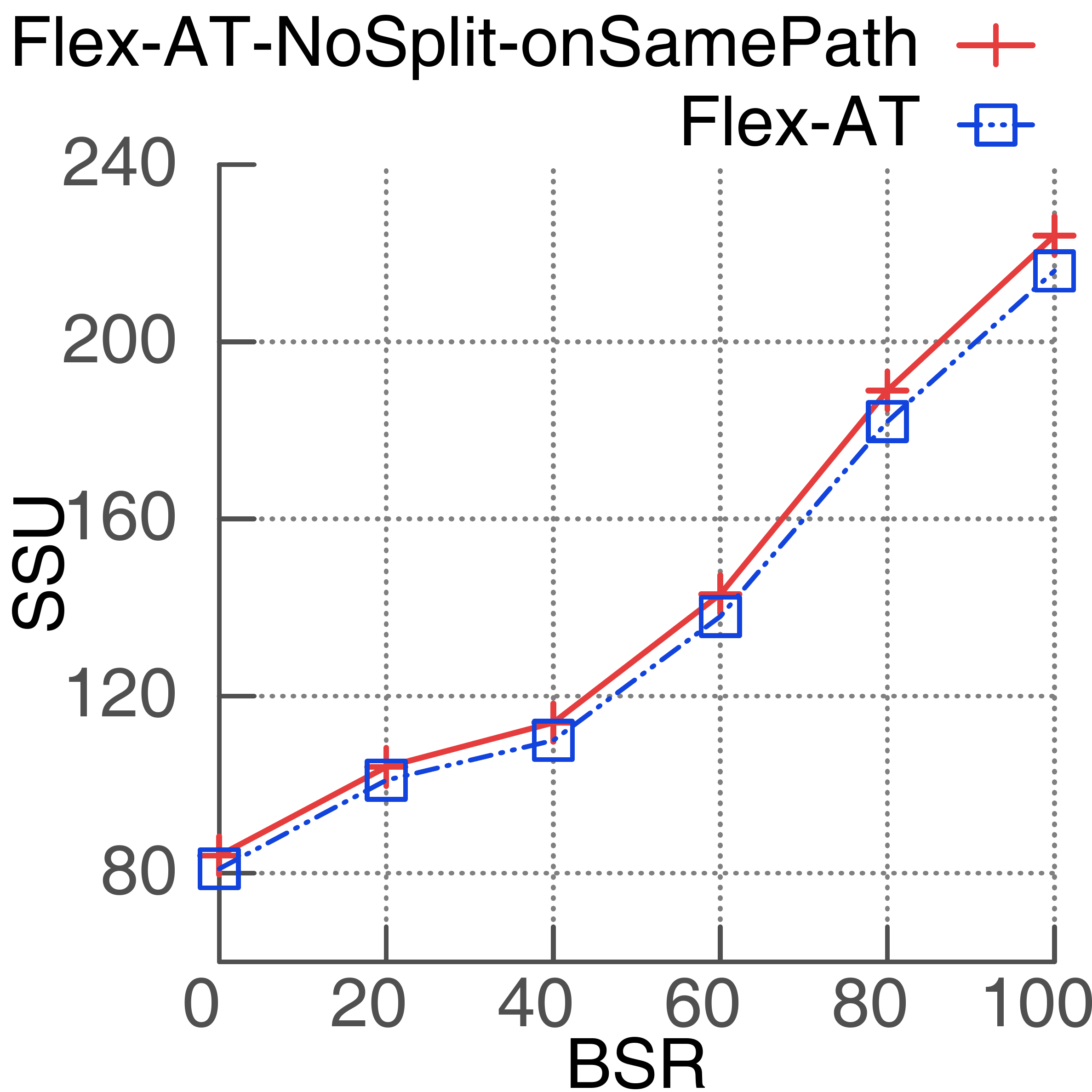}\label{fig_c:comp}}
	\caption{Analysis of our reliability model}
	\label{fig_c:lnr}
\end{figure}
\vspace{-0.5em}
\subsection{Results and Discussion}\label{sec_c:result}\vspace{-0.5em}
\subsubsection{Small-scale results with ILP formulation} \fig{fig_c:acratio} presents the impact of varying BSR on all performance metrics, and for all the compared variants, in Nobel Germany SN. Specifically, \fig{fig_c:acc} shows that SSU is the highest for the baseline variant, \textit{Fix-RT}, since it does not allow any flexibility in tuning transmission parameters or spectrum allocation. In contrast, \textit{Fix-AT} uses $30$\% less spectrum resources on average compared to \textit{Fix-RT} by only using tunable transponders with coarse-grained spectrum allocation. By combining flexible spectrum allocation with transponder tunability, \textit{Flex-AT} saves on average $57$\% and $27$\% spectrum compared to \textit{Fix-RT} and \textit{Fix-AT}, respectively. Lastly, \textit{Fix-RT} was infeasible for some problem instances at BSR $>80$ due to its inflexibility in transmission configuration tuning and spectrum allocation. 



\fig{fig_c:acc} also shows that SSU increases for increasing BSR in all considered cases. However, SSU rises very slowly (remaining within 10\% of the originally requested bandwidth\footnote{10\% additional bandwidth is due to the lack of fine-grained data rates that results in a small amount of over-provisioning (\eg 350Gbps is needed for BSR=40\% and demand=800Gbps, as 320Gbps is not a valid data rate).}) when BSR $\leq$ 40\%, while it steeply increases for  BSR $>$ 40\%. This can be explained by observing how the protection overhead (see plot in \fig{fig_c:bacc}) behaves for BSR up to 40\%. This behavior of the protection overhead for BSR $\leq$ 40\% confirms the intuition, already discussed in  \sect{sec:intro}, that low BSR can be guaranteed with low or no additional bandwidth if at least two disjoint SPaths can be used to map the splits of a VLink. In contrast, for BSR $>$ 40\%, more than two disjoint SPaths are needed for each VLink to reduce protection overhead. However, in our SN (Nobel Germany), either more than two disjoint SPaths do not exist for some pairs of SNodes, or the third and higher order disjoint SPaths become too long due to the sparse connectivity of the SN. Since long SPaths can support only low modulation formats, and thus require high number of spectrum slices, they are very unlikely to be selected in the optimal solution. 
This can be verified in \fig{fig_c:dist}, where it can be seen that the max. no. of disjoint SPaths used to embed a VLink is slightly larger than two for BSR $>$ 40\%, thus increasing protection overhead and, eventually, SSU. 

\fig{fig_c:split} shows that the max. no. of splits of a VLink is always larger than the max. no. of disjoint SPaths of a VLink, implying that there are multiple splits on the same SPath. Note also that \fig{fig_c:split} shows that \textit{Flex-AT} uses the highest number of splits to minimize the SSU as shown in \fig{fig_c:acc}. This is due to \textit{Flex-AT}'s smaller granularity in spectrum allocation that allows to use shorter SPaths even if the spectrum resources on these SPaths are fragmented.

To demonstrate how SN connectivity impacts spectral efficiency, \fig{fig_c:bacc_sn} shows SSU for \textit{Flex-AT} with different BSR requirements by varying LNR of Nobel Germany SN. \fig{fig_c:lnr} shows that SSUs decrease with increasing SN LNR even for the no protection case (0\% BSR). This stems from the fact that an SN with a larger LNR reduces the lengths and the number of intermediate hops of the SPaths between the same pair of SNodes, thus facilitating the use of more spectrum efficient transmission configurations along fewer hops. However, the gain in spectrum saving by increasing SN LNR is much higher for high BSR requirements. The additional gain is due to the use of three or more disjoint SPaths for mapping a VLink, thanks to the higher path diversity of a densely connected SN. 

Finally, \fig{fig_c:comp} demonstrates the benefit of our splitting model that enables \textit{Flex-AT} to save 3\%--4\% SSU compared to a case in which splitting is allowed over different SPaths, but not over the same SPath on different spectrum segments (\textit{Flex-AT-NoSplit-onSamePath}, e.g., the splitting model of \cite{goscien2016survivable}). Note that \textit{NoSplit-onSamePath} model renders problem instances with high BSR in \textit{Fix-RT} and \textit{Fix-AT} cases infeasible as spectrum fragmentation prohibits only one split per feasible SPath to occupy the required spectrum. In contrast, our model finds solutions in those cases by increasing number of splits per SPath even in fragmented situations as shown in \fig{fig_c:acc} and \fig{fig_c:split}. It is worth noting that our splitting model does not require any additional hardware investment, so the savings shown in \fig{fig_c:comp} are achieved at no additional cost.     

\begin{figure}[!h]
	\centering
	\subfigure[Small scale]{\includegraphics[trim=1.1cm 1cm 1.75cm 2.5cm,width=0.24\textwidth]{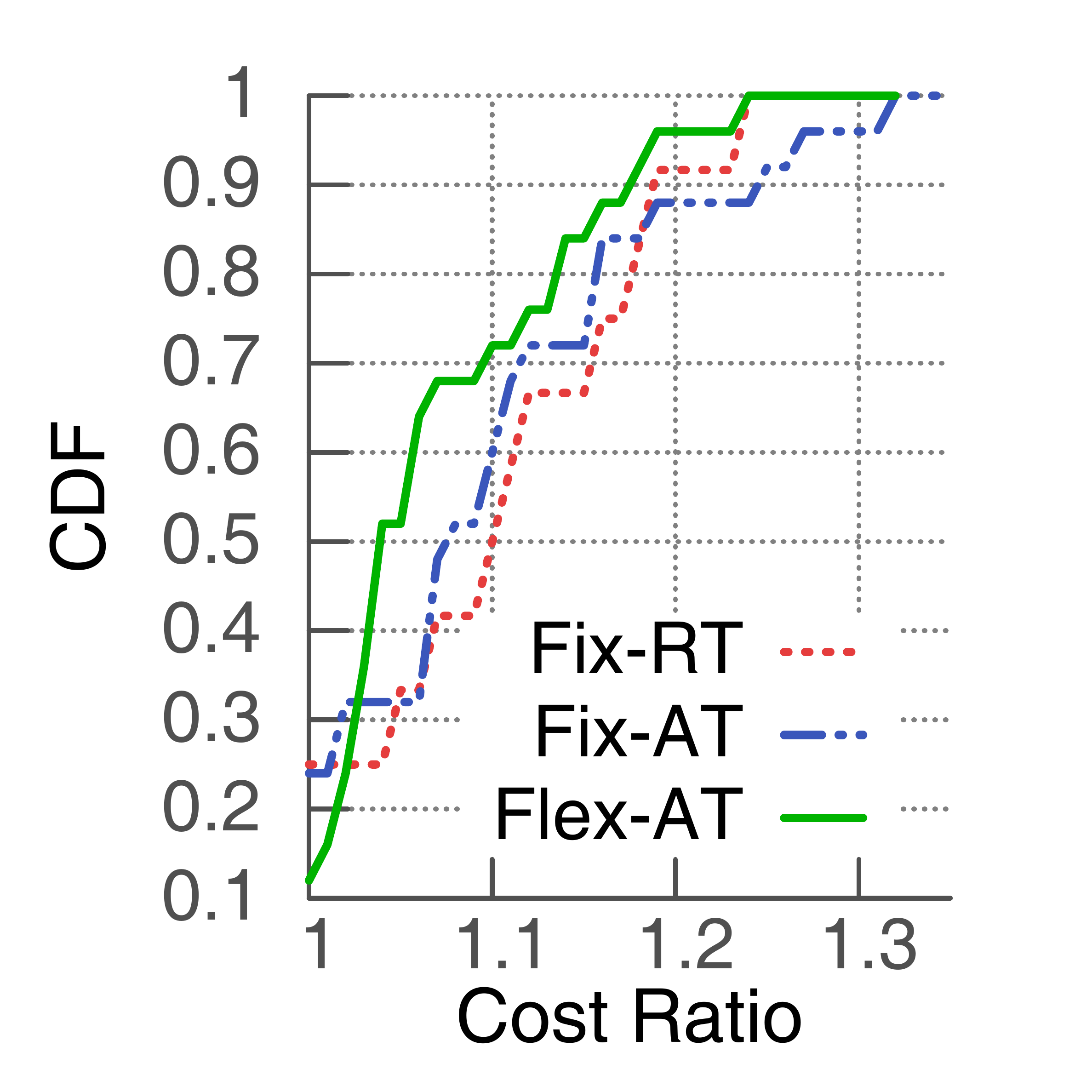}\label{fig_c:cdf}}
	\subfigure[Large scale]{\includegraphics[trim=2cm 1.25cm 1.15cm 2.5cm,width=0.24\textwidth]{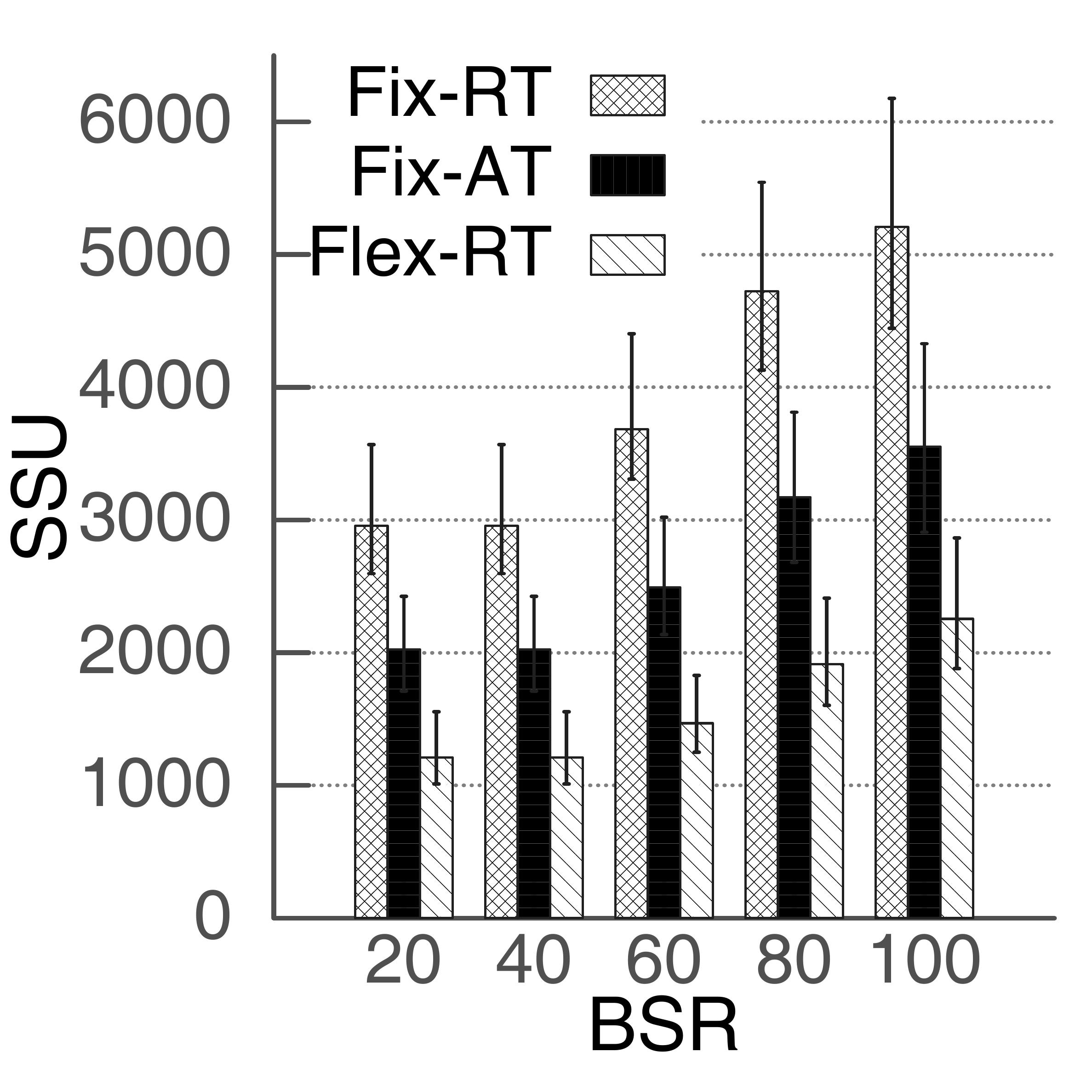}\label{fig_c:lar_cost}}
	\caption{Performance of our Heuristic Algorithm}
	\vspace{0em}
	\label{fig_c:large}
\end{figure}

\subsubsection{Performance of the Heuristic Algorithm}
As our cost function~\eqn{obj} is dominated by SSU, the cost ratio between heuristic and ILP gives a measure of how much additional resources are allocated by the heuristic. We present the cumulative distribution function (CDF) of cost ratios in~\fig{fig_c:cdf}. Over all instances, the heuristic incurs 12\%, 9\%, and 7\% additional cost on average compared to the ILP for \textit{Fix-RT}, \textit{Fix-AT}, and \textit{Flex-AT}, respectively, while executing 2 or 3 orders of magnitude faster. \fig{fig_c:lar_cost} shows SSU incurred by the heuristic algorithm in Germany50 SN. \fig{fig_c:lar_cost} confirms the conclusions observed in \fig{fig_c:acratio}, but with much larger VN, SN, and full 4THz spectrum resources on SLinks and by repeating our simulations over a much higher number of instances to achieve statistical confidence. Although not reported in this paper for space limitation, the heuristic finds a solution in large scale simulations in $\approx3$ minutes. 

%% file: conclusion.tex
\vspace{-1em}
\section{Conclusion}\label{sec:conclusion}	\vspace{-0.5em}
This paper addresses a fundamental problem for the slicing of 5G transport networks, \ie reliable VN embedding with dedicated protection in an EON. To reduce resource overbuild of dedicated protection, we exploit bandwidth squeezing and VLink demand splitting over multiple SPaths, while leveraging the flexibilities offered by an EON. Our novel splitting model not only provides the opportunity to split a VLink demand across multiple SPaths but also across multiple spectrum segments of an SPath. We present an ILP formulation to solve the problem optimally, and a heuristic solution to address the computational complexity of the ILP. Our simulations on realistic network topologies show that bandwidth squeezing and demand splitting allow to significantly reduce spectrum usage for providing dedicated protection, especially in the case of fully-flexible EON. Our evaluation also shows that the opportunity to have multiple splits over the same path allows to save additional spectrum compared to a model that does not allow splitting on the same path, and our proposed heuristic performs close to the optimal solution.

%% file: proof.tex
\appendices
\newpage
\section{\alg{alg:VLink-Ordering} optimality proof}
\begin{restatable}{theorem}{onlytheorem}
\alg{alg:VLink-Ordering} returns a VLink ordering with the minimum commonality index. 
\end{restatable}
\begin{proof}
Suppose VLink ordering $o$ which is generated by \alg{alg:VLink-Ordering} does not have the minimum commonality index, therefore, there exists a VLink ordering $o^*$ for which $w^{o} > w^{o^*}$. Let $\bar{e}_1^o, \bar{e}_2^o, ..., \bar{e}_{\bar{|E|}}^o$ and $\bar{e}_1^{o^*}, \bar{e}_2^{o^*}, ..., \bar{e}_{\bar{|E|}}^{o^*}$ denote the VLink ordering $o$ and $o^*$ respectively. Since $w^{o} \neq w^{o^*}$ there exists at least one index $i$ for  which $\bar{e}_j^{o}$ and $\bar{e}_j^{o^*}$ are not corresponding to the same VLink. Let $i_{max}$ be the maximum index with this condition. Since both VLink ordering $o$ and $o^*$ contain all the VLinks, there should be an index $j$ such that $\bar{e}_j^{o^*}$ corresponds to the same VLink as $\bar{e}_{i_{max}}^{o}$. We create a new VLink ordering $o^*_{1}$ by moving the $\bar{e}_j^{o^*}$ to the $i_{max}$th index in the $o^*$ VLink ordering (\fig{fig:create-new-order}).


\begin{figure}[!h]
	\centering
	\includegraphics[trim=1.1cm 1.3cm 1.75cm 1.5cm,width=0.45\textwidth]{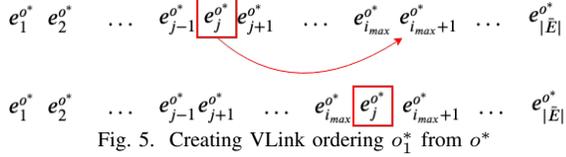}
	\caption{Creating VLink ordering $o^*_{1}$ from $o^*$}
	\label{fig:create-new-order}
\end{figure}

\begin{lemma}
\label{lemma1}
The commonality index of VLink ordering $o^*_{1}$ is less than or equal to the commonality index of the VLink ordering $o^*$ $(w^{o^*_{1}} \leq w^{o^*})$. \end{lemma}

\begin{proof}
We divide the proof of the lemma in two parts. At first, we prove \textit{i)} the effective commonality of all the VLinks (except $\bar{e}_j^{o^*}$) in VLink ordering $o^*_1$ is less than or equal to the effective commonality of the same VLink in VLink ordering $o^*$, i.e. $\forall \bar{e}_i \in \bar{E} - \{\bar{e}_j^{o^*}\}: w(\bar{e}_i)^{o^*_{1}} \leq w(\bar{e}_i)^{o^*}$. Then we prove \textit{ii)} $w(\bar{e}_j^{o^*})^{o^*_{1}} \leq w(\bar{e}_{i_{max}}^{o^*})^{o^*}$, hence, for each VLink in VLink ordering $o^*_1$ there exists at least one VLink in VLink ordering $o^*$ such that its effective commonality is greater or equal, i.e. $\forall \bar{e}_i \in \bar{E}, \exists \bar{e}_j: w(\bar{e}_i)^{o^*_{1}} \leq w(\bar{e}_j)^{o^*}$. From \eqn{eq:eq-link-weight}, the commonality index of a VLink ordering is equal to the maximum effective commonality of all VLinks in that ordering, therefore, $w^{o^*_1} \leq w^{o^*}$.

To prove \textit{i)}, we divide VLinks into 3 groups (\fig{fig:three-groups}):
\begin{enumerate}
    \item $A = \{\bar{e}_i^{o^*} | 1 \leq i < j\}$ which includes all the VLinks that come before $\bar{e}_j^{o^*}$ in both $o^*$ and $o^*_{1}$.
    \item $B = \{\bar{e}_i^{o^*} | j < i \leq i_{max} \}$ which contains all the VLinks that come after $\bar{e}_j^{o^*}$ in $o^*$, but before $\bar{e}_j^{o^*}$ in $o^*_1$.
    \item $C = \{\bar{e}_i^{o^*} | i_{max} < i \leq |\bar{E}| \}$ which includes the VLinks that come after $\bar{e}_j^{o^*}$ in both $o^*$ and $o^*_{1}$.
\end{enumerate}

\begin{figure}[!h]
	\centering
	\includegraphics[trim=1.1cm 2cm 1.75cm 1.2cm,width=0.45\textwidth]{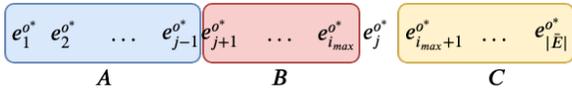}
	\caption{Dividing the VLinks into 3 groups to check their commonality indexes}
	\label{fig:three-groups}
\end{figure}

From \eqn{eq:eq-node-weight1}, we know that the effective commonality of each VLink only depends on the preceding VLinks. Hence, considering that the set of preceding VLinks is the same $\forall \bar{e}_i \in A \cup C$ in ordering $o^*$ and $o^*_1$, their effective commonalities are the same, i.e. $\forall \bar{e}_i \in A \cup C: w(\bar{e}_i^{o^*})^{o^*_{1}} = w(\bar{e}_i^{o^*})^{o^*}$. In addition, the set of preceding VLinks of $\bar{e}_i \in B$ in $o^*_{1}$ is the same as their set of preceding VLinks in $o^*$ minus $\bar{e}_j^{o^*}$. Therefore, $\forall \bar{e}_i \in B: w(\bar{e}_{i})^{o^*_1} = w(\bar{e}_{i})^{o^*} - w(\bar{e}_{i}, \bar{e}_j^{o^*}) \leq w(\bar{e}_{i})^{o^*}$.

For proving \textit{ii)}, we need to show that $w(\bar{e}_j^{o^*})^{o^*_{1}} \leq w(\bar{e}_{i_{max}}^{o^*})^{o^*}$. Recall that  \alg{alg:VLink-Ordering}, at each iteration, finds the VLink with minimum weight and puts it at the first empty place at the end of ordering (Line 4--7). Since the ordering $o$, which is generated by \alg{alg:VLink-Ordering}, is identical to ordering $o^*$ from index $i_{max}+1$ to $\bar{E}$, it means that VLink $\bar{e}_j^{o} = \bar{e}_{i_{max}}^{o^*}$ has the minimum weight, therefore $w(n_{\bar{e}_j^{o}}) \leq w(n_{\bar{e}_{i_{max}}^{o^*}})$ at $(n - i_{max} + 2)$-th step (to find the VLink that should be at index $i_{max}$). Since \eqn{eq:eq-node-weight} computes the effective commonality of the all the remaining VLinks assuming they will be at the last empty spots $w(n_{\bar{e}_j^{o}}) = w(\bar{e}_j^{o^*})^{o^*_{1}}$ and $w(n_{\bar{e}_{i_{max}}^{o^*}}) = w(\bar{e}_{i_{max}}^{o^*})^{o^*}$. So $w(\bar{e}_j^{o^*})^{o^*_{1}} \leq w(\bar{e}_{i_{max}}^{o^*})^{o^*}$.

By showing these two parts, we prove that $w^{o^*_{1}} \leq w^{o^*}$.
\end{proof}

We start from VLink ordering $o^*_{0} = o^*$ and at $i$th step we generate $o^*_{i}$ from $o^*_{i-1}$ using the same approach to generate $o^*_{1}$ from $o^*$. Since the maximum index for which $o^*_{1}$ and $o$ are different decreases after each step, finally, after $s \leq |\bar{E}|$ steps, $o = o^*_{s}$. By \cref{lemma1} we know that the commonality index does not increase after each step, i.e., $w^{o^*_{i}} \leq w^{o^*_{i-1}}$. Therefore $w^o = w^{o^*_{s}} \leq w^{o^*}$, which contradicts the initial assumption. This  means that \alg{alg:VLink-Ordering} returns an ordering with the minimum commonality index.
\end{proof}